\documentclass[journal]{IEEEtran}
\usepackage{indentfirst,amsmath,dblfloatfix}
\usepackage{color,rotating,pdflscape}
\usepackage{cite}
\usepackage{amsfonts,amsmath,amssymb,amsthm}
\usepackage{latexsym,amscd}
\usepackage{amsbsy,extarrows}
\usepackage{graphicx,multirow,bm}
\usepackage[table]{xcolor}
\usepackage{ulem}
\usepackage{extarrows}
\usepackage{tikz}
\usetikzlibrary{arrows,shapes}
\usetikzlibrary{shadows}

\newtheorem{Theorem}{Theorem}

\newtheorem{Definition}{Definition}
\newtheorem{Example}{Example}

\newtheorem{Remark}{Remark}
\newtheorem{Lemma}{Lemma}

\newtheorem{Proposition}{Proposition}

\begin{document}
\title{A Systematic Construction of MDS Codes\\ With Small Sub-packetization Level and\\ Near-Optimal Repair Bandwidth}
\author{Jie~Li,~\IEEEmembership{Member,~IEEE,} Yi Liu,~\IEEEmembership{Graduate~Student~Member,~IEEE,} and  Xiaohu Tang,~\IEEEmembership{Senior~Member,~IEEE}
\thanks{
The work of J. Li was supported in part by the National Science Foundation of China under Grant  61801176.   The work of Y. Liu and  X. Tang was supported in part by the National Natural Science Foundation of China under Grant 61871331
and Grant 61941106. This paper was presented in part at the 2019 IEEE
International Symposium on Information Theory.}
\thanks{J. Li is with the Department of Mathematics and Systems Analysis,
                    Aalto University, FI-00076 Aalto,  Finland, and also with the Hubei Key Laboratory of Applied Mathematics, Faculty of Mathematics and Statistics, Hubei University,
Wuhan 430062, China (e-mail: jie.0.li@aalto.fi, jieli873@gmail.com).}
\thanks{Y. Liu and X. Tang are with the Information Security and National Computing Grid Laboratory, Southwest Jiaotong University, Chengdu, 610031, China (e-mail: yiliu.swjtu@outlook.com; xhutang@swjtu.edu.cn).}
}
\maketitle

\begin{abstract}
In the literature, all the known high-rate MDS codes with the optimal repair bandwidth possess a significantly large sub-packetization level, which may prevent the codes to be implemented in practical systems. To build MDS codes with small sub-packetization level, existing constructions and theoretical bounds imply that one may sacrifice the optimality of the repair bandwidth.
Partly motivated by the work of Tamo \textit{et al.} (IEEE Trans. Inform. Theory, 59(3), 1597-1616, 2013), in this paper,  we present a   transformation that can greatly reduce the sub-packetization level of   MDS codes  with the optimal repair bandwidth with respect to  the same code length $n$. As applications of the transformation, four high-rate MDS codes having both small sub-packetization level and near-optimal repair bandwidth can be obtained, where three of them are  explicit and the required field sizes are   around or even smaller than the code length $n$. Additionally, we propose another explicit  MDS code which has a similar structure as that of the first resultant code obtained by the generic transformation, but can be built on a smaller finite field.
\end{abstract}


\section{Introduction}
\IEEEPARstart{I}{n} distributed storage systems such as Hadoop Distributed File System (HDFS) and Google File System (GFS), redundancy is imperative to ensure reliability. An attractive solution  is to call upon the maximum distance separable (MDS) codes, which provide the optimal tradeoff between fault tolerance and storage overhead. By distributing the codeword across distinct storage nodes, in the case of node failures, the missing data can be recovered from the data at some surviving nodes,  named helper nodes as well.  For this  scenario, one of the most important parameters is the \textit{repair bandwidth}, which is defined as the amount of data downloaded from the helper nodes to repair the failed node. Particularly, Dimakis \textit{et al.} \cite{Dimakis} derived a lower bound on the repair bandwidth of MDS codes,
 which motivated abundant recent research in coding for distributed storage \cite{Goparaju,invariant_subspace,product,repair-parity-zigzag,hadamard,Hadamard-strategy,Sasidharan-Kumar2,Zigzag,Long_IT,transform-ISIT,transform-IT,Barg1,Barg2,extend-zigzag,elyasi2019cascade,elyasi2018cascade,YiLiu,li2019systematic,balaji2018erasure}.

In the literature, most existing MDS codes with the repair bandwidth achieving the lower bound in \cite{Dimakis} are  a kind of array codes. A codeword of an $(n,k)$ array code is an $N\times n$ matrix, where the parameter $N$ is called the \textit{sub-packetization level} and $n$ is called the \textit{code length}.  When deploying an array code to a distributed storage system, a code symbol (i.e., a column) corresponds to a storage node. Then, an array code is said to have the \textit{MDS property} if any $k$ out of the $n$ columns of the matrix can recover the remaining $n-k$ columns.
It was proved in \cite{Dimakis} that the repair bandwidth $\gamma(d)$ of an $(n,k)$ MDS array code with  sub-packetization level $N$ should satisfy
\begin{equation}\label{Eqn_bound_on_gamma}
 \gamma(d)\ge \gamma^*(d) \triangleq\frac{d}{d-k+1}N,
\end{equation}
where $d$ ($k\le d\le n-1$) is the number of helper nodes.
An MDS array code is said to have the optimal repair bandwidth if  $\gamma(d)=\gamma^*(d)$,  i.e., the amount of data downloaded from each helper node is $\frac{N}{d-k+1}$.
In the particular case, when $d=n-1$, $\gamma^*(d)$ can be  reduced to the minimal value $\frac{n-1}{n-k}N$. Therefore, $d=n-1$ is the main concern in the most known  works \cite{repair-parity-zigzag,hadamard,Zigzag,Hadamard-strategy,extend-zigzag,Long_IT,transform-ISIT,transform-IT,invariant_subspace}. In this paper, we  also follow the same setting  and thus abbreviate $\gamma^*(n-1)$ to $\gamma^*$. Especially,
we focus on MDS array codes, and   refer to them as MDS codes for simplicity.

Up to now, various  MDS code constructions with the optimal repair bandwidth have been proposed,  among which some   notable works are \cite{Goparaju,product,hadamard,Sasidharan-Kumar2,invariant_subspace,Zigzag,Barg1,Barg2,YiLiu,transform-IT,elyasi2019cascade,elyasi2018cascade}.
However, in the high-rate regime,  all the known  $(n, k)$ MDS code constructions with the optimal repair bandwidth possess a significantly large  sub-packetization level $N$, i.e., $N\ge r^{n\over r+1}$ where $r=n-k$ \cite{Long_IT}.
In  \cite{Goparaju_bound}, it was shown that  for an $(n,k)$ MDS code with the optimal repair bandwidth, a  sub-packetization level $N$  being  exponential in the square root of $k$  is necessary. Very recently, Alrabiah and Guruswami \cite{alrabiah2019exponential}  further improved the lower bound on $N$ to being  exponential in   $k$ and they conjectured that the construction in \cite{Long_IT} with $N=r^{n\over r+1}$ is exactly tight. An MDS code with a larger sub-packetization level can lead to a reduced design space in
terms of various system parameters and  make management
of meta-data difficult. Moreover,  the implementation in practical systems is a big challenge \cite{Rawat}.

Existing constructions and  theoretical bounds imply that one may construct  high-rate MDS codes with small sub-packetization level by sacrificing the  optimality of the repair bandwidth.  In \cite{Rawat}, two  high-rate $(n,k)$ MDS codes with small sub-packetization level were presented. The first one can have a sub-packetization level as small as $N=r^\tau$  where $r=n-k$ and $\tau$ is a positive integer with $1\le \tau \le \lceil\frac{n}{r} \rceil-1$, while the repair bandwidth is no larger than $(1+\frac{1}{\tau})\gamma^*$.  However, the code is constructed over a significantly large finite field $\textbf{F}_q$ with $q>n^{(r-1)N+1}$, which may  hinder its deployment in practical systems.  The second MDS code  is obtained by combining an MDS code with the optimal repair bandwidth and another error-correcting code  with specific parameters. The proposed codes, therefore, rely on the existence of the latter, which may not always be available. For convenience, we refer to the two codes in \cite{Rawat} as RTGE code 1 and RTGE code 2 in this paper.
In \cite{Zigzag}, an $(n=sk'+2,k=sk')$ MDS code with sub-packetization level $2^{k'-1}$ and near-optimal repair bandwidth only for systematic nodes was proposed, which is termed duplication-zigzag code in this paper. In fact, the duplication-zigzag code is constructed based on $s$-duplication of the $(k'+2,k')$ zigzag code, but  can only support two parity nodes.

In this paper, we aim to construct high-rate MDS codes that have both small sub-packetization level and near-optimal repair bandwidth for general parameters  $n$ and $k$  over a small finite field $\textbf{F}_q$.  Partly motivated by the work in \cite{Zigzag}, we present a   transformation that can convert any   $(n'=k'+r, k')$ MDS  code with the optimal repair bandwidth that is defined in the parity-check matrix form into another $(n=k+r, k)$ MDS code with much longer code length. Specifically, the repair bandwidth of the new MDS code is upper bounded by $(1+\frac{r}{n'})\gamma^*$, but     the sub-packetization level is kept unchanged,   or equivalently the generic transformation can reduce the  sub-packetization level $N$ of the original codes with respect to the same code length $n$. By directly applying the generic transformation to several known  high-rate MDS codes with the optimal repair bandwidth, we get four high-rate $(n,k)$ MDS codes with both small sub-packetization level $N$ and near-optimal repair bandwidth,   three of which  are explicit and the required field sizes are  {around or smaller than  the code length $n$. Besides, we propose another new   MDS code    which has a similar structure as that of the first resultant code obtained by the generic transformation, but can be built on a smaller finite field.  The obtained MDS codes outperform the RTGE code 1 in \cite{Rawat} in terms of the field size, and   the first codes in both \cite{Barg1} and \cite{YiLiu} as well as the RTGE code 2 in \cite{Rawat}  in terms of the sub-packetization level. As a matter of convenience, we refer to the first  two codes in \cite{Barg1} respectively as YB code 1 and YB code 2, while referring to the first code in \cite{YiLiu} as the improved YB code 2 (since it is an improvement  of the YB code 2 in \cite{Barg1} with respect to the field size).

The remainder of the paper is organized as follows. Section II reviews some necessary preliminaries. Section
III proposes the generic transformation and
its asserted properties. Section IV demonstrates several  applications of the generic transformation,  three of which    are explicit. Section V presents another new explicit construction of high-rate MDS code over a small finite field that   has a small sub-packetization level,  near-optimal repair bandwidth, and the optimal update property. Section VI gives comparisons of   key parameters among the MDS codes proposed in this paper and some existing  notable  MDS  codes. Finally, Section VII concludes the study.

\section{Preliminaries}
In this section, we introduce some  preliminaries on high-rate MDS  codes, and  a series of special partitions for a given basis set.

\subsection{$(n,k)$ MDS codes}
Denote by $q$ a prime power and $\mathbf{F}_q$ the finite field with $q$ elements. For any two integers $a$ and $b$ with $b>a$, denote by $[a, b)$ the set $\{a, a+1, \ldots, b-1\}$. Let $\mathbf{f}_0, \mathbf{f}_1, \ldots, \mathbf{f}_{n-1}$ be the data stored across  a distributed storage system consisting of $n$ nodes   based on an $(n,k)$ MDS code, where $\mathbf{f}_i$ is a column vector of length $N$ over $\mathbf{F}_q$.  Throughout this paper, we consider   $(n,k)$ MDS codes that permit a definition in the following parity-check form:
\begin{equation}\label{Eqn parity check eq}
  \underbrace{\left(\hspace{-2mm}
        \begin{array}{cccc}
                A_{0,0} & A_{0,1} & \cdots & A_{0,n-1} \\
                A_{1,0} & A_{1,1} & \cdots & A_{1,n-1} \\
       \vdots & \vdots & \ddots & \vdots \\
            A_{r-1,0} & A_{r-1,1} & \cdots & A_{r-1,n-1}
            \end{array}
            \hspace{-2mm}\right)}_{A}
\left(\hspace{-2mm}\begin{array}{c}
              \mathbf{f}_0\\
              \mathbf{f}_1\\
             \vdots\\
         \mathbf{f}_{n-1}
            \end{array}
         \hspace{-2mm}\right)=\mathbf{0}_{rN},
\end{equation}
where $r=n-k\ge2$, $\mathbf{0}_{rN}$ denotes the zero column vector of length $rN$, and will be abbreviated as $\mathbf{0}$ in the sequel if its length is clear. The $rN\times nN$ block matrix $A$ in \eqref{Eqn parity check eq} is called the \textit{parity-check matrix} of the code, which can be written as
\begin{equation*}
A=(A_{t,i})_{t\in [0,r),i\in[0,n)}
\end{equation*}
to indicate the block entries.

For every  $t\in[0,r)$, by \eqref{Eqn parity check eq}, we have  $\sum\limits_{i=0}^{n-1}A_{t,i}\mathbf{f}_i=\mathbf{0}$, which contains $N$ linear equations. Particularly, we say that $\sum\limits_{i=0}^{n-1}A_{t,i}\mathbf{f}_i=\mathbf{0}$ is the $t$-th \textit{parity-check group}.

\subsection{The MDS property}
An $(n,k)$  MDS code defined by \eqref{Eqn parity check eq} possesses the MDS property that the source file can be reconstructed by connecting to any $k$ out of the $n$ nodes. That is,  any $r\times r$ sub-block matrix of $(A_{t,i})_{t\in [0,r),i\in[0,n)}$
is nonsingular \cite{Barg1}.

In particular, if
\begin{equation}\label{Eqn A power}
 A_{t,i}=A_i^{t}, ~t\in [0,r), ~i\in [0,n)
\end{equation}
for some matrices $A_i$ of order $N$, then we have the following result.

\begin{Lemma} [\cite{Barg1}]\label{Lemma pre MDS}
An $(n,k)$   code defined by \eqref{Eqn parity check eq}  and \eqref{Eqn A power} has the MDS property if
$A_iA_j=A_jA_i$ and $A_i-A_j$ is nonsingular for all $i,j\in [0,n)$ with $i\ne j$.
\end{Lemma}

\subsection{Repair}\label{sec:repair}

When repairing a failed node $i$ ($i\in[0,n)$) of an $(n,k)$ MDS   code,  denote by $\beta_{i,j}$ the amount of data downloaded from node $j$, where  $j\in [0,n)\backslash\{i\}$. In fact, the data downloaded from  helper node $j$ can be represented by $R_{i,j}\mathbf{f}_j$, where $R_{i,j}$ is a $\beta_{i,j} \times N$ matrix  of full rank.  Throughout this paper, $R_{i,j}$ is called  the \textit{repair matrix} of node $i$.

Clearly,  a failed node can be repaired if there are  $N$ linearly independent equations with respect to the  $N$ unknowns of $\mathbf{f}_i$. Specially, the $N$ equations should be chosen elaborately so that the interference in these equations can be cancelled by the downloaded data $R_{i,j}\mathbf{f}_j$ from the helper nodes $j\in [0,n)\backslash\{i\}$. In this paper, similar to that in \cite{YiLiu}, for convenience, we only consider the symmetric situation where appropriate $N/r$  linearly independent equations are acquired from each of the  $r$ parity-check groups, which are linear combinations of the corresponding $N$ parity-check equations.
Precisely, these $N/r$  linearly independent  equations  can be obtained by multiplying  the $t$-th parity-check group with an $N/r \times N$ matrix $S_{i,t}$ of full rank, where $S_{i,t}$ is called the \textit{select matrix}  in \cite{YiLiu}.
As a consequence, the following linear equations are available.
\begin{equation*}
\underbrace{\left(\begin{array}{c}
S_{i,0} A_{0,i} \\
S_{i,1} A_{1,i}\\
\vdots\\
S_{i,r-1} A_{r-1,i}
\end{array}\right)\mathbf{f}_i}_{\mathrm{useful ~data}}
+\sum_{j=0,j\ne i}^{n-1}\underbrace{\left(\begin{array}{c}
S_{i,0} A_{0,j}\\
S_{i,1} A_{1,j}\\
\vdots\\
S_{i,r-1} A_{r-1,j}
\end{array}\right)\mathbf{f}_j}_{\mathrm{interference ~by~}\mathbf{f}_j}=\mathbf{0},
\end{equation*}
thus regenerating node $i$ requires that
\begin{itemize}
\item [(i)] the coefficient matrix of the useful data  is of full rank, i.e.,
\begin{equation}\label{repair_node_requirement1 n-1}
\textrm{rank}(\left(
\begin{array}{c}
S_{i,0} A_{0,i}\\
S_{i,1} A_{1,i} \\
\vdots \\
S_{i,r-1} A_{r-1,i}
\end{array}
\right)) =N, \, i\in [0,n) ,
\end{equation}
\item [(ii)] the interference caused by $\mathbf{f}_j$  can  be determinable by the data $R_{i,j} \mathbf{f}_{j}$ downloaded from node $j$ for all $j\in [0,n)\backslash\{i\}$, i.e.,
\begin{equation*}\label{repair_node_requirement2 n-1}
  \mbox{rank}(\left(
  \begin{array}{c}
  R_{i,j}\\
  S_{i,0}A_{0,j}\\
  S_{i,1}A_{1,j}\\
  \vdots\\
  S_{i,r-1}A_{r-1,j}
  \end{array}
  \right))=\mbox{rank}\left(R_{i,j}\right),
\end{equation*}
for $i,j\in[0,n)$ with $i\ne j$,
which means that
\begin{equation}\label{repair_node_requirement3 n-1}
\textrm{rank} (\left(
\begin{array}{c}
R_{i,j} \\
S_{i,t} A_{t,j}
\end{array}
\right)) =\mbox{rank}(R_{i,j})
\end{equation}
for $i,j\in[0,n)$ with $i\ne j$, $t\in [0,r)$.
\end{itemize}

Then, the repair bandwidth of node $i$ is
\begin{equation}\label{Eqn_RB}
\gamma_i=\sum\limits_{j=0,j\ne i}^{n-1} \mathrm{rank}(R_{i,j})=\sum\limits_{j=0,j\ne i}^{n-1} \beta_{i,j}.
\end{equation}
As mentioned before,  a lower repair bandwidth of a node is desirable. According to \eqref{Eqn_bound_on_gamma}, if $\gamma_i=\gamma^*=(n-1){N\over r}$, then node $i$ is said to have the optimal repair bandwidth.  If $\gamma_i\le (1+\epsilon)\gamma^*=(1+\epsilon)(n-1){N\over r}$ for a small constant $\epsilon$, then node $i$ is said to have the \textit{near-optimal repair bandwidth} \cite{Rawat}.

In addition to the (near-) optimal repair bandwidth,  an $(n,k)$ MDS code is also preferred to have the \textit{optimal update} property, that is,  the minimum number of elements need to be updated  when an information element is changed.   In \cite{Barg1}, Ye and Barg showed that an $(n,k)$ MDS code defined in the form of \eqref{Eqn parity check eq} and \eqref{Eqn A power} has the optimal update property if all the block matrices of the parity-check matrix are diagonal.

\subsection{Partition of basis $\{e_0,\cdots,e_{N-1}\}$}\label{subsection:partition}

Assuming that $N=r^m$ for two integers $r$ and $m$ with $r,m\ge2$,  let $e_0,\cdots,e_{r^m-1}$ be a basis of $\mathbf{F}_q^{r^m}$. For example,
they can be simply set as the standard basis, i.e.,
\begin{equation*}
    e_i=(0,\cdots,0,1,0,\cdots,0),\,\,i\in [0, r^m),
\end{equation*}
with only the $i$th entry being nonzero.

In \cite{invariant_subspace}, a series of special  partitions of the set $\{e_0,\cdots,e_{r^m-1}\}$ is given for $r=2$. These  set partitions can be easily generalized
to  the  case of  $r\ge2$, which will  play an important role in our proposed  new  constructions.

For consistency, we follow the notation in  \cite{invariant_subspace} hereafter.
Given an integer $0\le a<r^m$, denote by $(a_0,\cdots,a_{m-1})$
its $r$-ary expansion, i.e., $a=\sum\limits_{j=0}^{m-1}r^{m-1-j}a_{j}$.
For $0\le i< m$ and $0\le t<r$, define a subset of $\{e_0,\cdots,e_{r^m-1}\}$ as
\begin{equation}\label{Eqn_Vt}
V_{i,t}=\{e_a|a_i=t, 0\le a< r^m\},
\end{equation}
where $a_i$ is the $i$th element in the $r$-ary expansion of $a$. Moreover, for $0\le t<r$, we define a special subset of $\{e_0,\cdots,e_{r^m-1}\}$ as
\begin{equation}\label{Eqn V*}
V_{*,t}=\{e_a|a_0+a_1+\cdots+a_{m-1}=t, 0\le a< r^m\},~
\end{equation}
where $a_0+a_1+\cdots+a_{m-1}$ is computed modulo $r$. This special subset  will be used in the MDS code construction in Section \ref{sec:C2C3}.

Straightforwardly, $|V_{i,t}|=r^{m-1}$, and  $\{V_{i,0},V_{i,1},\cdots,V_{i,r-1}\}$ is a partition of the set $\{e_0,\cdots,e_{r^m-1}\}$ for any $i\in [0,m)\cup \{*\}$.
Table \ref{example partition} gives two examples of the set partitions  defined in \eqref{Eqn_Vt} and \eqref{Eqn V*}.
{\small
\begin{table}[htbp]
\begin{center}
\caption{(a) and (b) denote the $m+1$  partitions of  the set $\{e_0,\cdots,e_{r^m-1}\}$    defined by \eqref{Eqn_Vt}  and \eqref{Eqn V*} for $m=3,r=2$, and $m=2,r=3$, respectively.}
\label{example partition}\begin{tabular}{|c|c|c|c|c|c|c|c|c|c|}
\hline $i$ & 0 & 1 & 2 & * &$i$ & 0 & 1 & 2&*\\
\hline \multirow{4}{*}{$V_{i,0}$ }&$e_0$&$e_0$&$e_0$&$e_0$&\multirow{4}{*}{$V_{i,1}$ }&$e_4$&$e_2$&$e_1$&$e_1$\\
  & $e_1$&$e_1$&$e_2$ &$e_3$&& $e_5$&$e_3$&$e_3$&$e_2$\\
  &$e_2$&$e_4$&$e_4$&$e_5$&& $e_6$&$e_6$&$e_5$&$e_4$\\
  &$e_3$&$e_5$&$e_6$&$e_6$&& $e_7$&$e_7$&$e_7$&$e_7$\\
\hline\multicolumn{8}{c}{\hspace{15mm}(A)}
\end{tabular}\\\vspace{5mm}\setlength{\tabcolsep}{4pt}
\begin{tabular}{|c|c|c|c|c|c|c|c|c|c|c|c|c}
\hline $i$ & 0 & 1 & *& $i$ & 0 & 1 & *&$i$ & 0 & 1&*\\
\hline \multirow{3}{*}{$V_{i,0}$ }&$e_0$&$e_0$&$e_0$&\multirow{3}{*}{$V_{i,1}$ }&$e_3$&$e_1$&$e_1$&\multirow{3}{*}{$V_{i,2}$ }&$e_6$&$e_2$&$e_2$\\
  & $e_1$&$e_3$&$e_5$&& $e_4$&$e_4$&$e_3$&& $e_7$&$e_5$&$e_4$\\
  & $e_2$&$e_6$&$e_7$&& $e_5$&$e_7$&$e_8$&& $e_8$&$e_8$&$e_6$\\
\hline\multicolumn{6}{c}{\hspace{43mm}(B)}
\end{tabular}
\end{center}
\end{table}
}

Based on the $m$ set partitions in \eqref{Eqn_Vt}, let us define
\begin{equation}\label{B3}
    V_{i+sm,t}=V_{i,t},~i\in [0,~m),~s\ge1,~\mbox{and}~t\in [0,~r).
\end{equation}

Further, for any $0\le i_1, i_2< sm$ and $i_1\not \equiv i_2\mbox{\ mod\ } m$, we define $V_{i_1,i_2,t_1,t_2}=V_{i_2,i_1,t_2,t_1}=V_{i_1,t_1}\cap V_{i_2,t_2}$, i.e.,
\begin{IEEEeqnarray*}{rCl}
V_{i_1,i_2,t_1,t_2}&=& V_{i_2,i_1,t_2,t_1}\\&=&\{e_a|a_{i_1}=t_1, ~a_{i_2}=t_2, ~a\in [0,~ r^m)\},
\end{IEEEeqnarray*}
where $0\le t_1,t_2<r$.
Then, we have
\begin{equation}\label{Eqn_B3}
V_{i_1,t_1}=V_{i_1,i_2,t_1,0}\cup\cdots \cup V_{i_1,i_2,t_1,r-1}.
\end{equation}
For the easy of notation,  we also denote by $V_{i_1,t_1}$ and $V_{i_1,i_2,t_1,t_2}$  the $r^{m-1}\times r^m$ and $r^{m-2}\times r^m$ matrices,
 whose rows are formed by vectors $e_i$
in their corresponding sets, respectively,  such that $i$  is sorted in ascending order.  For example, when $r=2$ and $m=3$, $V_{1,0}$ can be viewed as a $4\times 8$ matrix as follows
\begin{equation*}
V_{1,0}=\left(e_0^{\top}~ e_1^{\top} ~e_4^{\top}~ e_5^{\top}\right)^{\top},
\end{equation*}
where $\top$ represents the transpose operator.

\section{A generic transformation}\label{sec generic}

In this section, we present a generic transformation that can convert any MDS code with the optimal repair bandwidth defined in the form of \eqref{Eqn parity check eq} to a new MDS code with longer code length and near-optimal repair bandwidth.

{\textbf{A generic transformation: }The transformation can be performed through the following two steps.}
\\\textbf{Step 1. Choose  an $(n',k')$ MDS code  with the optimal repair bandwidth as the base code}

We choose an $(n',k')$ MDS code  in the form of \eqref{Eqn parity check eq}, with the optimal repair bandwidth  over a finite field containing at least $q'$ elements,  as the base code. Let $N$ denote its sub-packetization level, $r=n'-k'$, and let $(A'_{t,i})_{t\in[0,r), i\in[0,n')}$ denote its parity-check matrix while the $N/r\times N$ matrices $R'_{i,j}$ and $S'_{i,t}$,  $i,j\in[0,n')$ with $j\ne i$, $t\in[0,r)$, respectively denote the repair matrices and select matrices.\\
\textbf{Step 2.  Transform the base code  to the new MDS code}

Through the generic transformation, we intend to design a new $(n=k+r,k)$ MDS code over a certain finite field $\textbf{F}_q$ ($q>q'$) having arbitrary code length $n$ ($n>n'$) while maintaining the same sub-packetization level $N$.

The transition from the base code to the new MDS code is done by designing the   parity-check matrix, the repair matrices, and the select matrices  of the new MDS code from those of the base code
as follows.
\begin{IEEEeqnarray}{rCl}
A_{t,j}&=&x_{t,j}A'_{t,j\%n'},\label{Eqn general coding matrix}\\
R_{i,j}&=&\left\{
                      \begin{array}{ll}
                       R'_{i\% n',j\% n'}, &\mbox{\ \ if\ \ } j\not\equiv i \bmod n', \\
                        I_{N}, & \mbox{\ \ otherwise,\ \ }
                      \end{array}
                    \right.\label{Eqn general R}
\end{IEEEeqnarray}
and
\begin{equation}\label{Eqn general S}
S_{i,t}=S'_{i\% n',t}
\end{equation}
where $x_{t,j}\in \mathbf{F}_q\backslash\{0\}$, $t\in[0,r)$,
$i,j\in[0,n)~\mbox{with}~j\ne i$,
$\%$ denotes the modulo operation,  and $I_{N}$ denotes the identity matrix of order $N$, which will be abbreviated as $I$ in the sequel if its order is clear.

\begin{Remark}
For an $(n',k')$ MDS code defined over a finite field that contains at least $q'$ elements, it  can of course be defined over a larger finite field $\textbf{F}_q$  ($q>q'$). In the above generic transformation, the base code is then assumed to be defined over the same finite field  $\textbf{F}_q$  of the resultant new code.
\end{Remark}

Like many  MDS  codes in the literature,  the  MDS property of the resultant code  can be guaranteed by the  Combinatorial Nullstellensatz in \cite{Alon}.

\begin{Lemma}[Theorem 1.2 of \cite{Alon}]\label{Lemma Comb Null}
Let $\mathbf{F}_q$ be an arbitrary field, and $f=f(x_1,\cdots,x_n)$ be a polynomial in $\mathbf{F}_q[x_1,\cdots,x_n]$. Suppose that the degree  of $f$ is $\sum\limits_{i=1}^{n}t_i$, where each $t_i$ is a nonnegative integer, and the coefficient of $\prod\limits_{i=1}^{n}x_i^{t_i}$ in $f$ is nonzero. Then, if $S_1,\cdots,S_n$ are subsets of $\mathbf{F}_q$ with $|S_i|> t_i$,  there are $s_1\in S_1,\cdots,s_n\in S_n$ so that
 \begin{equation*}
  f(s_1,\cdots,s_n)\ne 0.
 \end{equation*}
\end{Lemma}

\begin{Theorem}\label{Thm general MDS}
The new $(n,k)$ code over $\mathbf{F}_q$ obtained by the generic transformation  can possess the MDS property if
\begin{itemize}
    \item [i)] $q>N{n-1\choose r-1}+1$,\footnote{Note that the field size required for the base code is $\ge q'$, therefore, $q$ should actually satisfy $q\ge \max\{q', N{n-1\choose r-1}+2\}$. However, the smallest field size required for any known explicit MDS code with the optimal repair bandwidth in the literature is far less than $N{n-1\choose r-1}+2$. So, we make an assumption here that $q'<N{n-1\choose r-1}+2$.} and
    \item [ii)] every block matrix $A'_{t,j}$ of the  parity-check matrix $(A'_{t,j})_{t\in[0,r), j\in[0,n')}$ of the base code    is nonsingular.
\end{itemize}
\end{Theorem}
\begin{proof}
The proof is given  in  Appendix \ref{sec:Appen1}.
\end{proof}

\begin{Remark}
To the best of our knowledge, there are only four classes of  MDS codes  with the optimal repair bandwidth  that are defined in parity-check matrix form, where the requirement in Theorem \ref{Thm general MDS}-ii)  can  be satisfied for two of them, i.e., the YB code 2 in \cite{Barg1} and    the improved YB code 2  in \cite{YiLiu},  while the remaining codes (i.e., the YB code 1 in \cite{Barg1} and the constructions in \cite{Barg2} and \cite{Sasidharan-Kumar2}) need a minor modification.  As a concrete example, the YB code 1 in \cite{Barg1}   satisfying this requirement will be   illustrated in  Section \ref{sec:C1}.
\end{Remark}

\begin{Theorem}\label{Thm gene band}
Every failed node of the new $(n,k)$ code obtained by the generic transformation can be regenerated by the repair matrices defined in \eqref{Eqn general R}, where the repair bandwidth   for node $i$ ($i\in [0, n)$) is
\begin{equation*}
  \gamma_i = \left\{
                      \begin{array}{ll}
                      (1+\frac{(\lceil\frac{n}{n'}\rceil-1)(r-1)}{n-1})\gamma^*, &\mbox{\ \ if\ \ } 0\le i\% n'<n\%n', \\
                       (1+\frac{(\lfloor\frac{n}{n'}\rfloor-1)(r-1)}{n-1})\gamma^*, & \mbox{\ \ otherwise}.
                      \end{array}
                    \right.
\end{equation*}
\end{Theorem}

\begin{proof}
Since the $(n',k')$ base code possesses the optimal repair bandwidth,  by \eqref{repair_node_requirement1 n-1} and \eqref{repair_node_requirement3 n-1}, we have
\begin{equation}\label{Eqn base full rank}
\textrm{rank}(\left(
\begin{array}{c}
S'_{i,0} A'_{0,i}\\
S'_{i,1} A'_{1,i} \\
\vdots \\
S'_{i,r-1}A'_{r-1,i}
\end{array}
\right))=N,~\mbox{for}~i\in [0,n'),
\end{equation}
and
\begin{equation}\label{Eqn base half rank}
\textrm{rank} (\left(
\begin{array}{c}
R'_{i,j} \\
S'_{i,t} A'_{t,j}
\end{array}
\right))
=N/r,~i,j\in [0,n')~\mbox{with}~i\ne j
\end{equation}
for $t\in[0,r)$.

For $i,j\in[0,n)$ with $j\ne i$, we rewrite $i$ and $j$ as $i=un'+i'$ and $j=vn'+j'$ such that $i',j'\in [0,n')$.
Firstly, we verify \eqref{repair_node_requirement1 n-1} for the new code.
By \eqref{Eqn general coding matrix} and \eqref{Eqn general S},
\begin{IEEEeqnarray*}{rCl}
&&\textrm{rank}(\left(
\begin{array}{c}
S_{i,0} A_{0,i}\\
S_{i,1} A_{1,i} \\
\vdots \\
S_{i,r-1} A_{r-1,i}
\end{array}
\right))\\&=&\textrm{rank}(\left(
\begin{array}{c}
S'_{i',0} A'_{0,i'}\\
S'_{i',1} A'_{1,i'} \\
\vdots \\
S'_{i',r-1} A'_{r-1,i'}
\end{array}
\right))\\&=&N,
\end{IEEEeqnarray*}
where the last equality follows from \eqref{Eqn base full rank}.

Next, we check \eqref{repair_node_requirement3 n-1}   for the new code.
When $i'\ne j'$,
\begin{IEEEeqnarray}{rCl}\label{Eqn_ge_RB1}
\nonumber\textrm{rank} (\left(
\begin{array}{c}
R_{i,j} \\
S_{i,t} A_{t,j}
\end{array}
\right))&=&
\textrm{rank} (\left(
\begin{array}{c}
R'_{i',j'} \\
S'_{i',t}A'_{t,j'}
\end{array}
\right))\\
\nonumber&=&N/r\\&=&\textrm{rank}
(R_{i,j})
, ~t\in[0,r),
\end{IEEEeqnarray}
where  the second and third equalities follows from \eqref{Eqn base half rank} and  \eqref{Eqn general R}, respectively.
When  $i'=j'$,  similarly, we have
\begin{IEEEeqnarray}{rCl}\label{Eqn_ge_RB2}
\nonumber\textrm{rank} (\left(
\begin{array}{c}
R_{i,j} \\
S_{i,t} A_{t,j}
\end{array}
\right))&=&
\textrm{rank} (\left(
\begin{array}{c}
 I \\
S'_{i',t} A'_{t,j'}
\end{array}
\right))\\
\nonumber&=&N\\
&=&\textrm{rank}(R_{i,j}), ~t\in[0,r).
\end{IEEEeqnarray}

Therefore, according to \eqref{Eqn_RB}, \eqref{Eqn_ge_RB1}, and \eqref{Eqn_ge_RB2}, the repair bandwidth of node $i$ is
\begin{IEEEeqnarray*}{rCl}
\gamma_i &=& \sum\limits_{j=0,j\ne i}^{n-1} \mbox{rank}(R_{i,j})\\&=&(n-1)\frac{N}{r}\\&&+\frac{(r-1)N}{r}|\{j:j\in [0, n)\backslash\{i\},~ j\equiv i \bmod n'\}|\\&=&
  \left\{
                      \begin{array}{ll}
                      (1+\frac{(\lceil\frac{n}{n'}\rceil-1)(r-1)}{n-1})\gamma^*, &\mbox{\ \ if\ \ } 0\le i\% n'<n\%n', \\
                       (1+\frac{(\lfloor\frac{n}{n'}\rfloor-1)(r-1)}{n-1})\gamma^*, & \mbox{\ \ otherwise},
                      \end{array}
                    \right.
\end{IEEEeqnarray*}
where  $\gamma^*=(n-1)\frac{N}{r}$ is the optimal value for the repair bandwidth. This finishes the proof.
\end{proof}

\begin{Remark}
In fact, any $(n',k')$ MDS code without the optimal repair bandwidth can also be chosen as the base code in the generic transformation. Its repair bandwidth is  $(n'-1)\beta$, i.e., a failed node can be regenerated by downloading an amount of $\beta$ symbols from  each surviving node.  Then the repair bandwidth of the resultant MDS code would be  upper bounded by $(1+{(\lceil\frac{n}{n'}\rceil-1)(N/\beta-1)\over (n-1)})(n-1)\beta$ according to   a similar analysis as the proof of Theorem \ref{Thm gene band}.
\end{Remark}

\section{MDS code constructions by directly applying the generic transformation}

In this section,  by directly applying the generic transformation in Section \ref{sec generic} respectively to the  $(n', k')$ YB codes 1 and 2  in \cite{Barg1},  the $(n', k')$  improved  YB code 2  in \cite{YiLiu},  and the counterpart of the long MSR code  \cite{Long_IT} in the parity-check form,  we get four MDS codes with small sub-packetization level.

\subsection{An $(n,k)$ MDS code $\mathcal{C}_1$ by applying the generic transformation  to  the  YB code 1 in \cite{Barg1}}\label{sec:C1}

The $(n', k')$ YB code 1 was defined in \cite{Barg1} in the form of \eqref{Eqn parity check eq} and \eqref{Eqn A power},  with the optimal update property and the sub-packetization level being $N=r^{n'}$ where $r=n'-k'$. More precisely, the  parity-check matrix  $(A'_{t,i})_{t\in[0,r), i\in[0,n')}$ of the $(n', k')$ YB code 1  satisfies $A'_{t,i}=(A'_i)^t$ and
\begin{equation}\label{Eqn_YB_code}
\left(
                     \begin{array}{c}
                       V_{i,0} \\
                       V_{i,1} \\
                       \vdots \\
                       V_{i,r-1}
                     \end{array}
                   \right)A'_i=\left(
                     \begin{array}{c}
                      \lambda_{i,0} V_{i,0} \\
                      \lambda_{i,1} V_{i,1} \\
                       \vdots \\
                      \lambda_{i,r-1} V_{i,r-1}
                     \end{array}
                   \right),
\end{equation}
where  $V_{i,0},V_{i,1},\cdots,V_{i,r-1}$ are defined in \eqref{Eqn_Vt}, $\{\lambda_{i,t}\}_{i\in[0, n'),t\in[0,r)}$ are $rn'$ distinct elements in a finite field containing at least $rn'$ elements,
the repair matrices and select matrices are defined by
\begin{equation*}\label{Eqn Hadmard RS}
R'_{i,j}=S'_{i,t}=V_{i,0}+V_{i,1}+\cdots+V_{i,r-1}
\end{equation*}
for $i,j\in[0,n')$ with $j\ne i$, $t\in[0,r)$.

From \eqref{Eqn_YB_code}, it is obvious that $A'_{i}$ is nonsingular if and only if $\{\lambda_{i,t}\}_{t\in[0,r)}$ are $r$ nonzero elements. In order to meet Theorem \ref{Thm general MDS}-ii), i.e.,  in order for matrices in \eqref{Eqn_YB_code} to be invertible,   we can add a restriction that $\{\lambda_{i,t}\}_{i\in[0, n'),t\in[0,r)}$ are $rn'$ nonzero elements when applying the generic transformation to YB code 1. Accordingly, the requirement of the  field size $q$ of YB code 1 is then only increased from $q\ge rn'$ to $q\ge rn'+1$, which can be easily satisfied as the resultant new code will be defined over a finite field with size larger than $rn'$.

\begin{Theorem}\label{Thm C1}
By choosing the $(n', k')$ YB code 1 as the base code for  the generic transformation in Section \ref{sec generic}, an $(n,k)$ MDS code $\mathcal{C}_1$  over $\mathbf{F}_q$ with
 $k=n-r$ and $q>N{n-1\choose r-1}+1$ can be obtained. Specifically, the sub-packetization level of the MDS code $\mathcal{C}_1$ is  $r^{n'}$ while its repair bandwidth  for node $i$ ($i\in [0, n)$) is
\begin{equation*}
  \gamma_i = \left\{
                      \begin{array}{ll}
                      (1+\frac{(\lceil\frac{n}{n'}\rceil-1)(r-1)}{n-1})\gamma^*, &\mbox{\ \ if\ \ } 0\le i\% n'<n\%n', \\
                       (1+\frac{(\lfloor\frac{n}{n'}\rfloor-1)(r-1)}{n-1})\gamma^*, & \mbox{\ \ otherwise}.
                      \end{array}
                    \right.
\end{equation*}
\end{Theorem}

For the MDS code $\mathcal{C}_1$ directly obtained by the generic transformation, the required field size is  relatively large and the construction is implicit. In the following, through a concrete assignment of the coefficients $x_{t,j}$, $t\in [0, r)$ and $j\in [0, n)$ in \eqref{Eqn general coding matrix}, we provide a solution to determine  the exact field size of the MDS code  $\mathcal{C}_1$,    which is much smaller than $N{n-1\choose r-1}+2$.

\begin{Theorem}\label{Thm MDS C1 explicit}
The field size $q$ of the  $(n,k)$ MDS code $\mathcal{C}_1$ can be reduced to
\begin{equation}\label{Eqn_Thm4_q}
q >\hspace{-1mm}\left\{\hspace{-2mm}
                      \begin{array}{ll}
                       rn'(\lceil\frac{n}{rn'}\rceil\hspace{-.5mm}-\hspace{-.5mm}1)\hspace{-.5mm}+\hspace{-.5mm}r(n\%n'), &\hspace{-2mm} \mbox{if~}0\hspace{-.5mm}<\hspace{-.5mm}n\% (rn')<n',\\
                       rn'\lceil\frac{n}{rn'}\rceil, &\hspace{-2mm}\mbox{otherwise},
                      \end{array}
                    \right.
\end{equation}
with $r\mid(q-1)$ by setting
\begin{equation}\label{Eqn C1 la assi}
\lambda_{i',t}=\delta^{t}c^{i'}
\end{equation} in \eqref{Eqn_YB_code} and
\begin{equation}\label{Eqn C1 x assi}
x_{t,i}=x_{i}^t=(c^{zn'}\delta^v)^t
\end{equation}
in \eqref{Eqn general coding matrix}
for $t\in [0,r)$, $i=zrn'+vn'+i'\in [0,n)$, $z\in [0, \lceil\frac{n}{rn'}\rceil)$, $v\in [0, r)$, and $i'\in [0, n')$, where $c$ is a primitive element of the finite field $\mathbf{F}_q$ and $\delta=c^{\frac{q-1}{r}}$, i.e.,  a primitive $r$-th root of unity  in the finite field $\mathbf{F}_q$.
\end{Theorem}

\begin{proof} Obviously, we only need to verify the  MDS property of the  code $\mathcal{C}_1$.
Note  from \eqref{Eqn C1 x assi} that $\mathcal{C}_1$ is  defined in the form of    \eqref{Eqn parity check eq} and \eqref{Eqn A power}, i.e.,
\begin{equation}\label{Eqn coding matrix C1}
A_{t,i}=x_{t,i}A'_{t,i'}=(c^{zn'}\delta^v A'_{i'})^t=A_i^t
\end{equation}
for $i=zrn'+vn'+i'$ and the matrix $A_i\triangleq c^{zn'}\delta^v A'_{i'}$.
Then, by Lemma \ref{Lemma pre MDS},   the  code $\mathcal{C}_1$    possesses the MDS property if $A_iA_j=A_jA_i$ and $A_i-A_j$ is nonsingular for all $i,j\in [0,n)$ with $i\ne j$.

Firstly, from \eqref{Eqn_YB_code} and \eqref{Eqn coding matrix C1}, it is seen that $A_i$ is diagonal for $i\in[0,n)$, then $A_iA_j=A_jA_i$ holds for any $i,j\in [0,n)$ with $i\ne j$.

Secondly, we show that $A_i-A_j$ is nonsingular for all $i,j\in [0,n)$ with $i\ne j$.
Let $i=z_0rn'+v_0n'+i'$ and $j=z_1rn'+v_1n'+j'$, where $i\ne j$, $z_0, z_1\in [0, \lceil\frac{n}{rn'}\rceil)$, $v_0, v_1\in [0, r)$, and $i', j'\in [0, n')$.

If $j\not\equiv i \bmod n'$, i.e., $i'\ne j'$,  then
{\small
\begin{IEEEeqnarray*}{rCl}
\hspace{0mm}\nonumber &&\hspace{0mm}\mbox{rank}(A_i-A_j)\\
\hspace{0mm}\nonumber &=&\hspace{0mm}\mbox{rank}(c^{z_0n'}\delta^{v_0}A'_{i'}-c^{z_1n'}\delta^{v_1}A'_{j'})\\
\hspace{0mm}\nonumber &=&\hspace{0mm}\mbox{rank}(\left(
                     \begin{array}{c}
                       V_{i',0} \\
                       \vdots \\
                       V_{i',r-1}
                     \end{array}
\right)(c^{z_0n'}\delta^{v_0}A'_{i'}-c^{z_1n'}\delta^{v_1}A'_{j'}))\\
\hspace{0mm}\nonumber &=&\hspace{0mm}\mbox{rank}(\left(
                     \begin{array}{c}
 V_{i',j',0,0}(c^{z_0n'}\delta^{v_0}A'_{i'}-c^{z_1n'}\delta^{v_1}A'_{j'}) \\
                       \vdots\\
 V_{i',j',0,r-1}(c^{z_0n'}\delta^{v_0}A'_{i'}-c^{z_1n'}\delta^{v_1}A'_{j'}) \\
                        \vdots\\
 V_{i',j',r-1,0}(c^{z_0n'}\delta^{v_0}A'_{i'}-c^{z_1n'}\delta^{v_1}A'_{j'}) \\
                       \vdots \\
 V_{i',j',r-1,r-1}(c^{z_0n'}\delta^{v_0}A'_{i'}-c^{z_1n'}\delta^{v_1}A'_{j'})
                     \end{array}
\right))\\
\hspace{0mm}\nonumber&=&\hspace{0mm}\mbox{rank}(\left(\hspace{-1mm}
                     \begin{array}{c}
 (c^{z_0n'}\delta^{v_0}\lambda_{i',0}-c^{z_1n'}\delta^{v_1}\lambda_{j',0}) V_{i',j',0,0} \\
                       \vdots\\
 (c^{z_0n'}\delta^{v_0}\lambda_{i',0}-c^{z_1n'}\delta^{v_1}\lambda_{j',r-1})V_{i',j',0,r-1} \\
                        \vdots\\
 (c^{z_0n'}\delta^{v_0}\lambda_{i',r-1}-c^{z_1n'}\delta^{v_1}\lambda_{j',0}) V_{i',j',r-1,0} \\
                       \vdots \\
 (c^{z_0n'}\delta^{v_0}\lambda_{i',r-1}-c^{z_1n'}\delta^{v_1}\lambda_{j',r-1})V_{i',j',r-1,r-1}
                     \end{array}
\hspace{-1mm}\right))\\
\hspace{0mm}\nonumber&=&\hspace{0mm}\mbox{rank}(\left(\hspace{-1mm}
                     \begin{array}{c}
                      (\delta^{v_0}c^{z_0n'+i'}-\delta^{v_1}c^{z_1n'+j'}) V_{i',j',0,0} \\
                       \vdots\\
                        (\delta^{v_0}c^{z_0n'+i'}-\delta^{v_1+r-1}c^{z_1n'+j'})V_{i',j',0,r-1} \\
                        \vdots\\
                      (\delta^{v_0+r-1}c^{z_0n'+i'}-\delta^{v_1}c^{z_1n'+j'}) V_{i',j',r-1,0} \\
                       \vdots \\
                       (\delta^{v_0+r-1}c^{z_0n'+i'}-\delta^{v_1+r-1}c^{z_1n'+j'})V_{i',j',r-1,r-1}
                     \end{array}
\hspace{-1mm}\right))
\end{IEEEeqnarray*}
}
where the first, third, fourth, and fifth equalities  follow from  \eqref{Eqn coding matrix C1},  \eqref{Eqn_B3}, \eqref{Eqn_YB_code}, and \eqref{Eqn C1 la assi}, respectively.
Thus, $\mbox{rank}(A_i-A_j)=N$ if and only if
\begin{equation}\label{Eqn C1 inequality equi}
\delta^{v_0+t_0-v_1-t_1}\ne c^{(z_1-z_0)n'+j'-i'} \mbox{~for~all~}t_0, t_1\in[0,r).
\end{equation}
Note that \eqref{Eqn C1 inequality equi} always holds, otherwise,
$$\delta^{v_0+t_0-v_1-t_1}= c^{(z_1-z_0)n'+j'-i'}$$ for  some $t_0, t_1\in[0,r)$.
Raising   both sides  to the power of $r$, by $\delta^r=1$ one then gets
\begin{equation}\label{Eqn_contradiction}
1=\delta^{r(v_0+t_0-v_1-t_1)}= c^{r\left((z_1-z_0)n'+j'-i'\right)}.
\end{equation}
In the following, we prove that \eqref{Eqn_contradiction} does not hold, i.e., $$0<|r\left((z_1-z_0)n'+j'-i'\right)|<q-1.$$

Clearly,
\begin{equation*}
0<|r\left((z_1-z_0)n'+j'-i'\right)|\le W
\end{equation*}
where $W= zrn'+rw$, $z=\lceil\frac{n}{rn'}\rceil-1$, $w=-1$ if $n\%(rn')=1$ (in this case $zrn'+w=n-2$ due to $j'-i'\ne 0$), $w=n\%n'-1$ if $1<n\%(rn')< n'$ (in this case $zrn'+w=n-1$), and  $w=n'-1$ else (in this case $zrn'+w< n-1$ unless $n\%(rn')=n'$) , i.e.,
\begin{equation*}
W\hspace{-1mm}=\hspace{-1mm}\left\{\hspace{-1mm}\begin{array}{ll}
rn'(\lceil\frac{n}{rn'}\rceil-1)-r,&\hspace{-1mm} \mbox{if}~n\%(rn')=1\\
rn'(\lceil\frac{n}{rn'}\rceil-1)+r(n\%n')-r, &\hspace{-1mm}\mbox{if}~1<n\%(rn')< n'\\
rn'\lceil\frac{n}{rn'}\rceil-r, &\hspace{-1mm} \mbox{else}
\end{array}\right.
\end{equation*}
which together with $r\mid(q-1)$ implies that \eqref{Eqn_contradiction} does not hold when \eqref{Eqn_Thm4_q} is
satisfied.

If $j\equiv i \bmod n'$, i.e., $i'=j'$,  then
\begin{IEEEeqnarray*}{rCl}
&&\mbox{rank}(A_i-A_j)\\&=&\mbox{rank}(c^{z_0n'}\delta^{v_0}A'_{i'}-c^{z_1n'}\delta^{v_1}A'_{j'})\\
\nonumber &=&\mbox{rank}((c^{z_0n'}\delta^{v_0}-c^{z_1n'}\delta^{v_1})A'_{i'}),
\end{IEEEeqnarray*}
therefore, $A_i-A_j$ is nonsingular if and only if
\begin{IEEEeqnarray}{rCl}\label{Eqn_C1_inequality2}
\nonumber&&c^{z_0n'}\delta^{v_0}-c^{z_1n'}\delta^{v_1}\\
\nonumber&=&c^{z_1n'+\frac{q-1}{r}v_1}\left(c^{(z_0-z_1)n'+\frac{q-1}{r}(v_0-v_1)}-1\right)\\
&\ne& 0
\end{IEEEeqnarray}
since $A'_{i'}$ is nonsingular. Note that $z_0,z_1\in [0, \lceil\frac{n}{rn'}\rceil)$, $v_0,v_1\in [0, r)$, and $(z_0, v_0)\ne (z_1, v_1)$ according to $i'=j'$ and $i\ne j$, then we have
\begin{IEEEeqnarray*}{rCl}
 0&<&|(z_0-z_1)n'+\frac{q-1}{r}(v_0-v_1)|\\&\le &\left(\lceil\frac{n}{rn'}\rceil-1\right)n'+\frac{q-1}{r}(r-1),
\end{IEEEeqnarray*}
thus \eqref{Eqn_C1_inequality2} holds if $q-1>\left(\lceil\frac{n}{rn'}\rceil-1\right)n'+\frac{q-1}{r}(r-1)$, i.e., $q>\left(\lceil\frac{n}{rn'}\rceil-1\right)rn'+r$ by combining $r\mid (q-1)$.

This finishes the proof   after combining the above analysis.
\end{proof}

In the following, we give a concrete example of the MDS code $\mathcal{C}_1$ according to Theorem \ref{Thm MDS C1 explicit}.

\begin{Example}\label{Ex1}
Let $n'=3$, $r=2$ and $n=12$, then the parity-check matrix of the $(12,10)$ MDS code $\mathcal{C}_1$ over $\mathbf{F}_{13}$ is defined through
\begin{IEEEeqnarray*}{rCl}
&&A_0=\left(
      \begin{array}{c}
        e_0 \\
        e_1 \\
        e_2 \\
        e_3 \\
        \delta e_4 \\
        \delta e_5 \\
        \delta e_6 \\
        \delta e_7
      \end{array}
    \right),~A_1=\left(
      \begin{array}{c}
        c e_0 \\
        c e_1 \\
        \delta c e_2 \\
        \delta c e_3 \\
        c e_4 \\
        c e_5 \\
        \delta c e_6 \\
        \delta c e_7
      \end{array}
    \right),~A_2=\left(
      \begin{array}{c}
        c^2 e_0 \\
        \delta c^2 e_1 \\
        c^2 e_2 \\
        \delta c^2 e_3 \\
        c^2 e_4 \\
        \delta c^2 e_5 \\
        c^2 e_6 \\
        \delta c^2 e_7
      \end{array}
    \right),\\[6pt]&&A_3=\left(
      \begin{array}{c}
         \delta e_0 \\
         \delta e_1 \\
         \delta e_2 \\
         \delta e_3 \\
          e_4 \\
         e_5 \\
         e_6 \\
         e_7
      \end{array}
    \right),~A_4=\left(
      \begin{array}{c}
        \delta c e_0 \\
        \delta c e_1 \\
        c e_2 \\
        c e_3 \\
        \delta c e_4 \\
        \delta c e_5 \\
        c e_6 \\
        c e_7
      \end{array}
    \right),~A_5=\left(
      \begin{array}{c}
        \delta c^2e_0 \\
        c^2 e_1 \\
        \delta c^2e_2 \\
        c^2 e_3 \\
        \delta c^2e_4 \\
        c^2 e_5 \\
        \delta c^2e_6 \\
        c^2 e_7
      \end{array}
    \right),\\[6pt]
&&A_6=\left(
      \begin{array}{c}
      c^3 e_0 \\
      c^3   e_1 \\
    c^3     e_2 \\
    c^3     e_3 \\
        \delta c^3 e_4 \\
        \delta c^3 e_5 \\
        \delta c^3 e_6 \\
        \delta c^3 e_7
      \end{array}
\right),~A_7=\left(
      \begin{array}{c}
        c^4 e_0 \\
        c^4 e_1 \\
        \delta c^4 e_2 \\
        \delta c^4 e_3 \\
        c^4 e_4 \\
        c^4 e_5 \\
        \delta c^4 e_6 \\
        \delta c^4 e_7
      \end{array}
\right),~A_{8}=\left(\hspace{-.5mm}
      \begin{array}{c}
         c^5 e_0 \\
        \delta c^5 e_1 \\
        c^5 e_2 \\
        \delta c^5 e_3 \\
        c^5 e_4 \\
        \delta c^5 e_5 \\
        c^5 e_6 \\
        \delta c^5 e_7
      \end{array}
   \hspace{-.5mm}\right),\\[6pt]&&A_{9}=\left(\hspace{-.5mm}
      \begin{array}{c}
         \delta c^3 e_0 \\
         \delta c^3 e_1 \\
         \delta c^3 e_2 \\
         \delta  c^3 e_3 \\
        c^3 e_4 \\
        c^3 e_5 \\
        c^3 e_6 \\
        c^3 e_7
      \end{array}
\hspace{-.5mm}\right),~A_{10}=\left(\hspace{-.5mm}
      \begin{array}{c}
         \delta c^4 e_0 \\
        \delta c^4 e_1 \\
        c^4 e_2 \\
        c^4 e_3 \\
        \delta c^4 e_4 \\
        \delta c^4 e_5 \\
        c^4 e_6 \\
        c^4 e_7
      \end{array}\hspace{-.5mm}\right),~A_{11}=\left(\hspace{-.5mm}
      \begin{array}{c}
       \delta c^5e_0 \\
        c^5 e_1 \\
        \delta c^5e_2 \\
        c^5 e_3 \\
        \delta c^5e_4 \\
        c^5 e_5 \\
        \delta c^5e_6 \\
        c^5 e_7
      \end{array}\hspace{-.5mm}\right),
\end{IEEEeqnarray*}
where $c=2$ and $\delta=c^6=-1$.

To save space, we only give the repair matrices and select matrices of node 0, which are
\begin{equation*}
  R_{0,j}=\left\{
                      \begin{array}{ll}

                        I, &\mbox{\ \ if\ \ } j=3,6,9, \\
                        \left(
                          \begin{array}{c}
                            e_0+e_4 \\
                            e_1+e_5 \\
                             e_2+e_6 \\
                             e_3+e_7
                          \end{array}
                        \right), & \mbox{\ \ otherwise,\ \ }
                      \end{array}
                    \right.
\end{equation*}
and \begin{equation*}
      S_{0,0} =S_{0,1}= \left(
                          \begin{array}{c}
                            e_0+e_4 \\
                            e_1+e_5 \\
                             e_2+e_6 \\
                             e_3+e_7
                          \end{array}
                        \right).
    \end{equation*}
\end{Example}

\begin{Theorem}\label{C1-update}
The MDS code   $\mathcal{C}_1$ has the optimal update property.
\end{Theorem}

\begin{proof}
Note that all the  block matrices of the parity-check matrix of the MDS code $\mathcal{C}_1$  are diagonal. By the definition of  the optimal update property and the arguments in \cite{Barg1}, we conclude that the MDS code  $\mathcal{C}_1$ has the optimal update property.
\end{proof}

\subsection{Two $(n,k)$ MDS codes $\mathcal{C}_2$ and $\mathcal{C}_3$ by applying the generic transformation respectively to the YB code 2 in \cite{Barg1} and the improved YB code 2  in \cite{YiLiu}}\label{sec:C2C3}

For consistency, we borrow the notation in \cite{Barg1} and \cite{YiLiu}  in what follows.
Let $N=r^{n'-1}$ where $r=n'-k'$. For any  $a\in [0,N)$ with $(a_0,a_1,\cdots,a_{n'-2})$ being its $r$-ary expansion,
define
\begin{equation}\label{Eqn ai}
a(i,u)=(a_0,\cdots,a_{i-1},u,a_{i+1},\cdots,a_{n'-2})
\end{equation}
and
\begin{IEEEeqnarray}{rCl}\label{Eqn aij}
\nonumber \hspace{-7mm}&&a(i,j,u,v)\\\hspace{-7mm}&=&(a_0,\cdots,a_{i-1},u,a_{i+1},\cdots,a_{j-1},v,a_{j+1},\cdots,a_{n'-2}),
\end{IEEEeqnarray}
 where $0\le i<j< n'-1$ and $u,v\in[0,r)$.

For the $(n', k')$ YB code 2 in \cite{Barg1} and the $(n', k')$ improved YB code 2 in \cite{YiLiu}, both of them are defined  in the form of \eqref{Eqn parity check eq} and \eqref{Eqn A power}  with the sub-packetization level
$N$. More precisely,
the parity-check matrix  $(A'_{t,i})_{t\in[0,r), i\in[0,n')}$ of the $(n', k')$ YB code 2 in \cite{Barg1} is defined by  $A'_{t,i}=(A'_i)^t$ and
\begin{equation*}\label{Eqn_Barg_CodingM}
A'_i=\left\{\begin{array}{ll}
\sum\limits_{a=0}^{N-1}\lambda_{i,a_i}e_a^{\top}e_{a(i,a_i+1)}, &i\in [0,n'-1),\\
I, & i=n'-1,
\end{array}\right.
\end{equation*}
where
\begin{equation*}
\lambda_{i,a_i}=\left\{
\begin{array}{ll}
c^{i+1}, & \textrm{if $a_i =0$},\\
1, & \textrm{otherwise},
\end{array}
\right.
\end{equation*}
with $c$ being a primitive element of a   finite field with size larger than $n'$.
While the parity-check matrix  $(A'_{t,i})_{t\in[0,r), i\in[0,n')}$ of the $(n', k')$ improved YB code 2   in \cite{YiLiu} is defined by  $A'_{t,i}=(A'_i)^t$ and
\begin{equation}\label{Eqn_Liu_CodingM}
A'_i=\left\{\begin{array}{ll}
\sum\limits_{a=0}^{N-1}\lambda_{i,a}e_a^{\top}e_{a(i,a_i+1)}, &  i\in [0,n'-1),\\
I, & i=n'-1,
\end{array}\right.
\end{equation}
where
\begin{equation}\label{Eqn C3 lam}
\lambda_{i,a}=\left\{
\begin{array}{ll}
c, & \textrm{if $\sum\limits_{t=0}^{i}a_t = 0$},\\
1, & \textrm{otherwise},
\end{array}
\right.
\end{equation}
with $c$ being a primitive element of   a   finite field $\mathbf{F}_q$ with $(q-1)\nmid (r-1)$.

The YB code 2 in \cite{Barg1} and the improved YB code 2   in \cite{YiLiu} have the same
repair matrices and select matrices, which are respectively defined by
\begin{equation*}
R'_{i,j}=\left\{
\begin{array}{ll}
V_{i,0}, &\mathrm{if~}i\in[0,n'-1),\\
V_{*,0},&\mathrm{if~}i=n'-1,
\end{array}
\right.
\end{equation*}
and
\begin{equation*}\label{Eqn Liu S}
S'_{i,t}=\left\{
\begin{array}{ll}
V_{i,0}, &\mathrm{if~}i\in [0,n'-1),\\
V_{*,r-t},&\mathrm{if~}i=n'-1,
\end{array}
\right.
\end{equation*}
where $V_{i,0}$, $V_{*,0}$ and $V_{*,r-t}$ are defined in \eqref{Eqn_Vt} and \eqref{Eqn V*}.

By directly applying the generic transformation in Section \ref{sec generic}, we have the following result.
\begin{Theorem}\label{Thm C2}
Respectively choosing the  $(n', k')$ YB code 2 in \cite{Barg1} and the $(n', k')$ improved YB code 2 in \cite{YiLiu}  as the base code for the generic transformation in Section \ref{sec generic},  two $(n,k)$ MDS codes $\mathcal{C}_2$ and $\mathcal{C}_3$ over $\mathbf{F}_q$ with
$k=n-r$ and $q>N{n-1\choose r-1}+1$ can be obtained.  Particularly, for  both the MDS codes $\mathcal{C}_2$ and $\mathcal{C}_3$,  the sub-packetization level is   $r^{n'-1}$ while the repair bandwidth    for node $i$ ($i\in [0, n)$) is
\begin{equation*}
  \gamma_i = \left\{
                      \begin{array}{ll}
                      (1+\frac{(\lceil\frac{n}{n'}\rceil-1)(r-1)}{n-1})\gamma^*, &\mbox{\ \ if\ \ } 0\le i\% n'<n\%n', \\
                       (1+\frac{(\lfloor\frac{n}{n'}\rfloor-1)(r-1)}{n-1})\gamma^*, & \mbox{\ \ otherwise}.
                      \end{array}
                    \right.
\end{equation*}
\end{Theorem}



In the following, by  a concrete assignment of the coefficients $x_{t,j}$, $t\in [0, r)$ and $j\in [0, n)$ in \eqref{Eqn general coding matrix}, we provide a solution to determine the exact field sizes of the MDS codes  $\mathcal{C}_2$ and $\mathcal{C}_3$,  which are much smaller than $N{n-1\choose r-1}+2$. Hereafter, we only derive the  values of  $x_{t,j}$, $t\in [0, r)$ and $j\in [0, n)$ in \eqref{Eqn general coding matrix}  for the MDS code $\mathcal{C}_3$ in detail, while for MDS code $\mathcal{C}_2$, we just give the results but omit the analysis since it is similar to that of the  MDS code $\mathcal{C}_3$.

\begin{Theorem}\label{Thm MDS C2}
The field size $q$ of the  MDS code $\mathcal{C}_2$ can be reduced to $q>r\lceil{n'\over r}\rceil(\lceil{n\over n'}\rceil-1)+n'$ by setting $x_{t,i}=x_{i}^t=c^{\lfloor {i\over n'}\rfloor \lceil{n'\over r} \rceil t}$ in \eqref{Eqn general coding matrix}  for $t\in [0,r)$ and $i\in [0,n)$, where $c$ is a primitive element of $\mathbf{F}_q$.
\end{Theorem}


Before proving the result on $\mathcal{C}_3$,
 we  first introduce some results related to
the  parity-check matrix  (see \eqref{Eqn_Liu_CodingM}) of the $(n', k')$ improved YB code 2 in \cite{YiLiu}.

\begin{Lemma}[Lemma 2, \cite{YiLiu}] \label{lem commu}
For any $i,j\in [0,n')$ with $i\ne j$, $A'_iA'_j=A'_jA'_i$, where $A'_i$ and $A'_j$ are defined  in \eqref{Eqn_Liu_CodingM}.
\end{Lemma}

\begin{Lemma}[Lemma 3, \cite{YiLiu}] \label{lem lambda}
 For any $a\in [0,N)$ and $i, j\in [0, n'-1)$,
\begin{itemize}
\item [(i)] $\prod\limits_{t=0}^{r-1}\lambda_{i,a(i,j,a_i-t,a_j+t+l)}=c$ for $j>i$;
\item [(ii)] $\prod\limits_{t=0}^{r-1}\lambda_{j,a(i,j,a_i-t,a_j+t+l)}=1 \mbox{~or~} c^r$ for $j>i$;
\item [(iii)] $\prod\limits_{t=0}^{r-1}\lambda_{j,a(j,a_j+t)}=c$ for $j \geq 0$,
\end{itemize}
where $l\in [0,r)$ is a constant,  $c$ is a primitive element of $\mathbf{F}_q$, $a(i,j,u,v)$  and $\lambda_{i,a}$ are respectively defined in \eqref{Eqn aij} and \eqref{Eqn C3 lam}.
\end{Lemma}

\begin{Lemma}[Lemma 4,  \cite{YiLiu}]\label{lem AX}
 For any $i\in [0,n'-1)$ and $X=\sum\limits_{a=0}^{N-1}x_ae_a^{\top}\in \mathbf{F}_q^N$,
$A'_iX=\sum\limits_{a=0}^{N-1}\lambda_{i,a}x_{a(i,a_{i}+1)}e_a^{\top}$ where $A'_i$ is defined  in \eqref{Eqn_Liu_CodingM}.
\end{Lemma}

\begin{Theorem}\label{Thm MDS C3}
The field size $q$ of the $(n, k)$ MDS code $\mathcal{C}_3$ can be reduced to $q>\lceil\frac{n}{n'}\rceil$ with $q$ being odd if  $r$ is even, and $q>r\lceil\frac{n}{n'}\rceil$ otherwise,
 by setting
\begin{equation}\label{Eqn C3 x assi}
x_{t,i}=x_{i}^t=c^{\lfloor {i\over n'}\rfloor t}
\end{equation}
in \eqref{Eqn general coding matrix} for $t\in [0,r)$ and $i\in [0,n)$, where $c$ is a primitive element of $\mathbf{F}_q$.
\end{Theorem}

\begin{proof}
Still, we only need to verify the  MDS property of the  code $\mathcal{C}_3$.
It is seen from  \eqref{Eqn C3 x assi} that the  code $\mathcal{C}_3$ is  defined in the form of    \eqref{Eqn parity check eq} and \eqref{Eqn A power} with
\begin{equation}\label{Eqn_C3_PCM}
  A_{t, i}= A_{i}^t=(c^{\lfloor {i\over n'}\rfloor }A'_{i\%n'})^t, ~t\in [0, r).
\end{equation}
That is
\begin{equation}\label{Eqn coding matrix C5}
A_{un'+n'-1}=c^uA'_{n'-1}=c^uI, \mbox{~for} ~u\in[0,~\lfloor\frac{n}{n'}\rfloor),
\end{equation}
and
\begin{equation}\label{Eqn coding matrix C51}
A_{un'+i'}=c^uA'_{i'}=\sum\limits_{a=0}^{N-1}c^u\lambda_{i',a}e_a^{\top}e_{a(i',a_{i'}+1)},
\end{equation}
for $u\in[0,~\lceil\frac{n}{n'}\rceil)$ and $i'\in [0,~n'-1)$ with  $un'+i'<n$.
According to Lemma \ref{Lemma pre MDS},   the  code $\mathcal{C}_3$    possesses the MDS property if $A_iA_j=A_jA_i$ and $A_i-A_j$ is nonsingular for all $i,j\in [0,n)$ with $i\ne j$.

First, by Lemma \ref{lem commu},  \eqref{Eqn coding matrix C5} and  \eqref{Eqn coding matrix C51}, we easily see that $A_iA_j=A_jA_i$ holds for any $i,j\in [0,n)$ with $i\ne j$.

Next, we  show that $A_i-A_j$ is nonsingular. Note that $A_i-A_j$ being nonsingular is equivalent to saying that for any $X=\sum\limits_{a=0}^{N-1}x_ae_a^{\top}$, $(A_i-A_j)X=\mathbf{0}$ implies $X=\mathbf{0}$. In the following, we analyze it through three cases.
For $i,j\in [0, n)$ with $i\ne j$, let us rewrite $i=un'+i'$ and $j=vn'+j'$ for some $u,v\in [0, \lceil\frac{n}{n'}\rceil)$ and $i',j'\in [0,n')$, where $(u,i')\ne (v,j')$.

\textbf{Case 1:} If $i\equiv j ~\bmod n'$,  i.e., $i'=j'$ and $u\ne v$,  then by \eqref{Eqn_C3_PCM}, we have
\begin{equation*}
  A_i-A_j=(c^u-c^v)A'_{i'}=c^v(c^{u-v}-1)A'_{i'},
\end{equation*}
which is nonsingular since $0<|u-v|\le \lceil\frac{n}{n'}\rceil-1<q-1$.

\textbf{Case 2:} If $i\not\equiv j ~\bmod n'$, $i'\ne n'-1$, and $j'\ne n'-1$, then  by Lemma \ref{lem AX}, we have
 \begin{IEEEeqnarray*}{rCl}
\nonumber &&(A_i-A_j)X\\\nonumber&=&(c^uA'_{i'}-c^vA'_{j'})X\\
\nonumber&=&\sum\limits_{a=0}^{N-1}(c^u\lambda_{i',a}x_{a(i',a_{i'}+1)}e_a^{\top}-c^v\lambda_{j',a}x_{a(j',a_{j'}+1)}e_a^{\top})\\
&=&\mathbf{0}
\end{IEEEeqnarray*}
if and only if
\begin{equation*}
c^u\lambda_{i',a}x_{a(i',a_{i'}+1)}-c^v\lambda_{j',a}x_{a(j',a_{j'}+1)}=0, ~a\in [0,N),
\end{equation*}
which is equivalent to
\begin{IEEEeqnarray}{rCl}\label{Eqn case2 e2}
\nonumber \hspace{-6mm}x_a &=&\frac{c^v\lambda_{j',a(i',a_{i'}-1)}}{c^u\lambda_{i',a(i',a_{i'}-1)}}x_{a(i',j',a_{i'}-1,a_{j'}+1)}\\
\hspace{-6mm}&=&
{\prod_{t=0}^{r-1}c^v\lambda_{j',a(i',j',a_{i'}-t,a_{j'}+t-1)}\over \prod_{t=0}^{r-1}c^u\lambda_{i',a(i',j',a_{i'}-t,a_{j'}+t-1)}}x_a, ~a\in [0,N).
\end{IEEEeqnarray}

Applying Lemma \ref{lem lambda} to \eqref{Eqn case2 e2}, if $j'>i'$, we get
\begin{equation*}
  (c^{rv}-c^{ru+1})x_a =c^{rv}(1-c^{ru-rv+1})x_a=0,
\end{equation*}
or
\begin{equation*}
  (c^{rv+r-1}-c^{ru})x_a =c^{ru}(c^{rv-ru+r-1}-1)x_a=0,
\end{equation*}
otherwise, we have
\begin{equation*}
  (c^{rv+1}-c^{ru})x_a =c^{ru}(c^{rv-ru+1}-1)x_a=0,
\end{equation*}
or
\begin{equation*}
(c^{rv}-c^{ru+r-1})x_a =c^{rv}(1-c^{ru-rv+r-1})x_a=0.
\end{equation*}
If $r$ is even, then $ru-rv+1$, $rv-ru+r-1$, $rv-ru+1$, and $ru-rv+r-1$ is odd, thus $$c^{ru-rv+1}, c^{rv-ru+r-1}, c^{rv-ru+1}, c^{ru-rv+r-1}\ne 1$$ when $q$ is odd; Otherwise, for any {\small $$W\in \{|ru-rv+1|, |rv-ru+r-1|, |rv-ru+1|, |ru-rv+r-1|\},$$}we have
\begin{equation*}
  0<W\le r\lceil\frac{n}{n'}\rceil-1<q-1
\end{equation*}
when $q>r\lceil\frac{n}{n'}\rceil$, i.e., $$c^{ru-rv+1}, c^{rv-ru+r-1}, c^{rv-ru+1}, c^{ru-rv+r-1}\ne 1$$ when $q>r\lceil\frac{n}{n'}\rceil$.
Hence,  if $q$ is odd and $r$ is even, or $q>r\lceil\frac{n}{n'}\rceil$ and $r$ is odd, we have that $$(c^{rv}-c^{ru+1})(c^{rv+r-1}-c^{ru})(c^{rv+1}-c^{ru})(c^{rv}-c^{ru+r-1})\ne 0,$$ thus $x_a=0$ for all $a\in [0,N)$, i.e., $X=0$.
Then, $A_i-A_j$ is nonsingular.

\textbf{Case 3:} If $i\not\equiv j ~\bmod n'$ and either $i'=n'-1$ or $j'=n'-1$, W.L.O.G., assuming that $i'=n'-1$, then $j'\ne n'-1$. Similar to Case 2, we have
\begin{equation*}
x_a=x_a\prod_{t=0}^{r-1}{c^v\lambda_{j',a(j',a_{j'}+t)}\over c^u}, a\in [0,N),
\end{equation*}
which in conjunction with Lemma \ref{lem lambda}, we have $$(c^{rv+1}-c^{ru})x_a=0$$ for all $a\in [0,N)$.  This implies that $x_a=0$ for all $a\in [0,N)$ by a similar analysis as in Case 2, i.e., $X=\mathbf{0}$.
Thus, $A_i-A_j$ is nonsingular.

Collecting the above three cases, we finish the proof.
\end{proof}

Let us see to  what extent  the field size $q$ of the $(n, k)$ MDS code $\mathcal{C}_3$ can be reduced  by Theorem
\ref{Thm MDS C3}.  For example, when $n'=12$, $r=3$, and  $n=24$. According to Theorem \ref{Thm MDS C3}, we can set  $x_{t,i}=x_{i}^t=2^{\lfloor {i\over 12}\rfloor t}$
in \eqref{Eqn general coding matrix} over $\mathbf{F}_7$ for $t\in [0,3)$ and $i\in [0,24)$, where $2$ is a primitive element of $\mathbf{F}_7$. Whereas,  by Theorem \ref{Thm C2}, the existence of the MDS code $\mathcal{C}_3$ requires a finite field with size larger than $4\times 10^7$.

\subsection{An $(n,k)$ MDS code $\mathcal{C}_4$ obtained by applying the generic transformation to a newly constructed MDS code $\mathcal{C}'_4$}

In this section, by using the approach of \cite{Barg1},
we first  construct an $\left(n'=\left(r+1\right)m,k'=n'-r\right)$ MDS code $\mathcal{C}'_4$  with sub-packetization level $r^m$,  and then   propose an $(n,k)$ MDS code  $\mathcal{C}_4$ with small sub-packetization level by  applying the generic transformation to the   code $\mathcal{C}'_4$. In fact, the   code $\mathcal{C}'_4$  can be viewed as an extension of the $(n'=rm,k'=r(m-1))$ MDS code in \cite{Barg2} with a longer code length.
Besides, $\mathcal{C}'_4$ in parity-check form can also be regarded as the counterpart of the   $(n'=k'+r,k'=(r+1)m)$  long minimum storage regenerating (MSR) code \cite{Long_IT} in systematic form. For simplicity, we call  $\mathcal{C}'_4$ the long code in this paper.
 In the following, we give the parity-check matrix, repair matrices and select matrices of the long code $\mathcal{C}'_4$.

The parity-check matrix $(A'_{t,i'})_{t\in [0,r),i'\in [0, n')}$  of the $(n'=(r+1)m,k'=n'-r)$  long  code $\mathcal{C}'_4$ satisfies
\begin{equation}\label{Eqn A=yB}
A'_{t,i'}=y_{t,i'}B'_{t,i'}
\end{equation}
and \eqref{Eqn pc C4'} in the next page,
\begin{figure*}[hb]
\hrulefill
\begin{equation}\label{Eqn pc C4'}
\left(\begin{array}{c}
         V_{i',0}\\
         V_{i',1}\\
         \vdots\\
         V_{i',r-1}
       \end{array}\right)B'_{t,i'}=\left\{\begin{array}{ll}
         \left(\begin{array}{c}
          \lambda_{i',0}^tV_{i',0}+\sum\limits_{u=1}^{r-1}(\lambda_{i',0}^t-\lambda_{i',u}^t)V_{i',u}\\
           \lambda_{i',1}^tV_{i',1}\\
           \vdots\\
           \lambda_{i',r-1}^tV_{i',r-1}
         \end{array}\right), & \mbox{if~}0 \le i' <m,\\
          \left(\begin{array}{c}
           \lambda_{i',0}^tV_{i',0}\\
           \lambda_{i',1}^tV_{i',1}+\sum\limits_{u=0,u \ne 1}^{r-1}(\lambda_{i',1}^t-\lambda_{i',u}^t)V_{i',u}\\
           \lambda_{i',2}^tV_{i',2}\\
           \vdots\\
           \lambda_{i',r-1}^tV_{i',r-1}
         \end{array}\right), & \mbox{if~}m \le i' <2m,\\
         ~~~~~~~~~~~~~~~~~~~~~~~~~~\vdots &
        ~~~~~~~~ \vdots\\
         \left(\begin{array}{c}
           \lambda_{i',0}^tV_{i',0}\\
           \vdots\\
           \lambda_{i',r-2}^tV_{i',r-2}\\
           \lambda_{i',r-1}^tV_{i',r-1}+\sum\limits_{u=0}^{r-2}(\lambda_{i',r-1}^t-\lambda_{i',u}^t)V_{i',u}
         \end{array}\right), & \mbox{if~}(r-1)m \le i' <rm,\\
         \left(\begin{array}{c}
           \lambda_{i',0}^tV_{i',0}\\
           \lambda_{i',1}^tV_{i',1}\\
           \vdots\\
           \lambda_{i',r-1}^tV_{i',r-1}
         \end{array}\right), & \mbox{if~}rm \le i' <(r+1)m,
       \end{array}\right.
\end{equation}
\end{figure*}
where $y_{t,i'}, \lambda_{i',u}\in \mathbf{F}_{q'}\backslash \{0\}$ for $i'\in [0, n')$ and $t, u\in [0,r)$,  $V_{i', 0}, \ldots, V_{i',r-1}$  are respectively defined by \eqref{Eqn_Vt} for $i'\in [0, m)$ and \eqref{B3} for $i'\in [m, n')$, i.e.,
\begin{IEEEeqnarray}{rCl}\label{Eqn reint A}
\nonumber \hspace{-7mm}&&V_{i',v}B'_{t,i'}\\
\hspace{-7mm}&=&\left\{\hspace{-2mm}
    \begin{array}{ll}
        \lambda_{i',v}^tV_{i',v}+\hspace{-1mm}\sum\limits_{u=0,u \ne v}^{r-1}\hspace{-1mm}( \lambda_{i',v}^t\hspace{-1mm}-\hspace{-1mm}\lambda_{i',u}^t)V_{i',u}, &\mbox{if\ }\lfloor{i'\over m}\rfloor=v,\\
        \lambda^{t}_{i',v}V_{i',v}, &\mbox{otherwise,}
    \end{array}\right.
\end{IEEEeqnarray}
for $i' \in [0,n')$ and $v, t\in [0,r)$.  The repair matrices and select matrices of the $(n',k')$ MDS code $\mathcal{C}'_4$ are respectively defined by
     \begin{equation}\label{eqn the definition of repair matrix of base long MDS code}
     R'_{i',j'}=S'_{i',t}=\left\{\begin{array}{ll}
         V_{i',\lfloor \frac{i'}{m}\rfloor}, & \mbox{if~} 0 \le i' < rm, \\
         \sum\limits_{u=0}^{r-1}V_{i',u}, & \mbox{if~} rm \le i' < n',
       \end{array}\right.
     \end{equation}
for $j'\in [0, n')\backslash\{i'\}$ and $t\in [0, r)$.

Obviously, $B'_{t,i'}$ is nonsingular for $t\in [0,r)$ and $i'\in [0, n')$ according to \eqref{Eqn pc C4'}.
Then we have the following result.
\begin{Theorem}\label{Thm_C4'_MDS}
The code $\mathcal{C}'_4$ has the MDS property over $\mathbf{F}_{q'}$ if $q' > N{n'-1\choose r-1}+1$.
\end{Theorem}
\begin{proof}
It can be proven similar to that of Theorem \ref{Thm general MDS}.
\end{proof}

\begin{Theorem}\label{Thm_C4'_RB}
The code $\mathcal{C}'_4$ has the optimal repair bandwidth if $\lambda_{i',0},\lambda_{i',1},\cdots,\lambda_{i',r-1}$ are $r$ distinct elements in $\mathbf{F}_{q'}$ for any $i' \in [0,n')$.
     \end{Theorem}
     \begin{proof}
The proof is given in Appendix \ref{sec:Appen2}.
\end{proof}
Based on the long code $\mathcal{C}'_4$, we have the following result by directly applying the generic transformation.

\begin{Theorem}\label{Thm C4 result1}
By applying the generic transformation in Section \ref{sec generic} to   the  $(n', k')$ long code $\mathcal{C}'_4$, an $(n,k)$ MDS code $\mathcal{C}_4$ over $\mathbf{F}_q$ with
$k=n-r$ and $q>N{n-1\choose r-1}+1$ can be obtained. Specifically, the sub-packetization level of the MDS code $\mathcal{C}_4$ is   $r^{n'\over {r+1}}$ while its repair bandwidth   for node $i$ ($i\in [0, n)$) is
\begin{equation*}
  \gamma_i = \left\{
                      \begin{array}{ll}
                      (1+\frac{(\lceil\frac{n}{n'}\rceil-1)(r-1)}{n-1})\gamma^*, &\mbox{\ \ if\ \ } 0\le i\% n'<n\%n', \\
                       (1+\frac{(\lfloor\frac{n}{n'}\rfloor-1)(r-1)}{n-1})\gamma^*, & \mbox{\ \ otherwise}.
                      \end{array}
                    \right.
\end{equation*}
\end{Theorem}
In what follows, we present a solution to determine the exact field size of the MDS code $\mathcal{C}_4$ for the case of $r=2$, which is much smaller than $N{n-1\choose r-1}+2$.

By \eqref{Eqn general coding matrix} and \eqref{Eqn A=yB}, the parity-check matrix $(A_{t,i})_{t\in [0,r),i\in [0, n)}$ of the $(n,k)$ MDS code $\mathcal{C}_4$ satisfies
\begin{equation}\label{Eqn pc C4}
    A_{t,i}=x_{t,i}A'_{t,i\%n'}=x_{t,i}y_{t,i\%n'}B'_{t,i\%n'}=z_{t,i}B'_{t,i\%n'},
\end{equation}
where
\begin{equation*}
    z_{t,i}=x_{t,i}y_{t,i\%n'},\,\,t\in [0,r),\,\,i\in [0, n).
\end{equation*}
Then we have the following result.

\begin{Theorem}\label{Thm C4 result2}
       When $r=2$, the field size $q$ of the $(n, k)$ MDS code $\mathcal{C}_4$ can be reduced to
\begin{equation}\label{Eqn_Thm11_q}
  q>\hspace{-1mm}\left\{\hspace{-2mm}
                      \begin{array}{ll}
                       2m(\lceil\frac{n}{n'}\rceil-1)+2(n\% n'), & \mbox{if\ }0<n\% n'<m,\\
                       2m\lceil\frac{n}{n'}\rceil, &\mbox{otherwise},
                      \end{array}
                    \right.
\end{equation}
by setting
       \begin{equation}\label{Eqn z}
         z_{t,i}=c^{2mt\lfloor\frac{i}{n'}\rfloor}
       \end{equation}
for $t=0,1$, $i \in [0,n)$ and
       \begin{IEEEeqnarray}{rCl}
\label{C4r=2coeff1} &&\lambda_{i',0}=\lambda_{i'+m,0}=\lambda_{i'+2m,1}=c^{2i'},\\
\label{C4r=2coeff2} &&\lambda_{i',1}=\lambda_{i'+m,1}=\lambda_{i'+2m,0}=c^{2i'+1},
       \end{IEEEeqnarray}
       in \eqref{Eqn pc C4'} for $i' \in [0,m)$, where $n'=3m$ and $c$ is a primitive element of $\mathbf{F}_q$.
     \end{Theorem}
     \begin{proof}
According to \eqref{Eqn pc C4'},   the code $\mathcal{C}_4$ has the MDS property if and only if any $2 \times 2$ sub-block matrix of
       \begin{equation*}
         \left(\hspace{-2mm}\begin{array}{cccc}
          A_{0,0} & A_{0,1} & \cdots & A_{0,n-1}\\
           A_{1,0} & A_{1,1} & \cdots & A_{1,n-1}
         \end{array}\hspace{-2mm}\right)\hspace{-1mm}=\hspace{-1mm}\left(\hspace{-2mm}\begin{array}{cccc}
            I & I & \cdots & I\\
           A_{1,0} & A_{1,1} & \cdots & A_{1,n-1}
         \end{array}\hspace{-2mm}\right)
       \end{equation*}
       is nonsingular, i.e, $A_{1,i}-A_{1,j}$ is nonsingular for any $i,j \in [0,n)$ with $i \ne j$. Let us rewrite $i=un'+i'$ and $j=vn'+j'$ for some $u,v \in [0, \lceil\frac{n}{n'}\rceil)$ and $i',j' \in [0,n')$, where $(u,i') \ne (v,j')$. In the following, we analyze the nonsingularity of $A_{1,i}-A_{1,j}$ in the following 6 cases according to \eqref{Eqn pc C4}-\eqref{C4r=2coeff2}.

       \textbf{Case 1:} When $0 \le i'=j'<3m$, then
       \begin{IEEEeqnarray*}{rCl}
         && \mbox{rank}\left(
        A_{1,i} - A_{1,j}
         \right)\\&=&
         \mbox{rank}\left(z_{1,i}B'_{1,i'} - z_{1,j}B'_{1,j'}\right)
         \\&=&\mbox{rank}\left((z_{1,i}-z_{1,j})B'_{1,i'}\right)\\
         &=&\mbox{rank}((c^{2mu}-c^{2mv})B'_{1,i'})\\
         &=&N\\
         &\Leftrightarrow& c^{2mv}\left(c^{2m(u-v)}-1\right) \ne 0,
       \end{IEEEeqnarray*}
which always holds since $$0<|2m(u-v)|\le 2m\lceil\frac{n}{n'}\rceil-2m<q-1.$$

       \textbf{Case 2:} When $0 \le i'< j'<m$, then
       \begin{IEEEeqnarray}{rCl}
  \nonumber      && \mbox{rank}(A_{1,i} - A_{1,j})\\
  \nonumber&=&\mbox{rank}(z_{1,i}B'_{1,i'} - z_{1,j}B'_{1,j'})\\
  \nonumber&=&\mbox{rank}(
         \left(\begin{array}{c}
           V_{i',j',0,0}\\V_{i',j',0,1}\\V_{i',j',1,0}\\V_{i',j',1,1}
         \end{array}\right)(z_{1,i}B'_{1,i'} - z_{1,j}B'_{1,j'}))\\
   \nonumber      &=&\mbox{rank}(
         \left(\begin{array}{c}
           V_{i',j',0,0}(z_{1,i}B'_{1,i'} - z_{1,j}B'_{1,j'})\\
           V_{i',j',0,1}(z_{1,i}B'_{1,i'} - z_{1,j}B'_{1,j'})\\
           V_{i',j',1,0}(z_{1,i}B'_{1,i'} - z_{1,j}B'_{1,j'})\\
           V_{i',j',1,1}(z_{1,i}B'_{1,i'} - z_{1,j}B'_{1,j'})
         \end{array}\right))\\
  \nonumber       &=&\mbox{rank}(\left(\begin{array}{l}
           (z_{1,i}\lambda_{i',0}-z_{1,j}\lambda_{j',0})V_{i',j',0,0}\\\hspace{8mm}+z_{1,i}(\lambda_{i',0}-\lambda_{i',1})V_{i',j',1,0}\\\hspace{15mm}-z_{1,j}(\lambda_{j',0}-\lambda_{j',1})V_{i',j',0,1}\\
           (z_{1,i}\lambda_{i',0}-z_{1,j}\lambda_{j',1})V_{i',j',0,1}\\\hspace{8mm}+z_{1,i}(\lambda_{i',0}-\lambda_{i',1})V_{i',j',1,1}\\
            (z_{1,i}\lambda_{i',1}-z_{1,j}\lambda_{j',0})V_{i',j',1,0}\\\hspace{8mm}-z_{1,j} (\lambda_{j',0}-\lambda_{j',1})V_{i',j',1,1}\\
           (z_{1,i}\lambda_{i',1}-z_{1,j}\lambda_{j',1})V_{i',j',1,1}
           \end{array}\right))\\
    \nonumber     &=&N\\
     \nonumber     &\Leftrightarrow& z_{1,i}\lambda_{i',a}-z_{1,j}\lambda_{j',b}\ne 0 \mbox{~for~all~}a, b=0,1,\\
 \nonumber      &\Leftrightarrow& c^{2m(u-v)+2(i'-j')+a-b}-1\ne 0\mbox{~for~all~}a, b=0,1,
       \end{IEEEeqnarray}
which is equivalent to
\begin{equation}\label{Eqn-C4-Case2}
  0<|2m(u-v)+2(i'-j')+a-b|<q-1,~ a, b=0,1.
\end{equation}
Obviously,
\begin{equation*}
0<|2m(u-v)+2(i'-j')+a-b|\le W
\end{equation*}
where $W=2mz+2w+1$, $z=\lceil\frac{n}{n'}\rceil-1$, $w=n\%n'-1$ if $0<n\%n'<m$ and $w=m-1$ otherwise, i.e.,
\begin{equation*}
W=\left\{\hspace{-1mm}\begin{array}{ll}
2m\left(\lceil\frac{n}{n'}\rceil-1\right)+2(n\%n')-1,& \mbox{if}~0<n\%n'<m,\\
2m\lceil\frac{n}{n'}\rceil-1, & \mbox{otherwise.}
\end{array}\right.
\end{equation*}
Therefore, \eqref{Eqn-C4-Case2} holds  if \eqref{Eqn_Thm11_q} is satisfied.

\textbf{Case 3:} When $m \le i' < j'<2m$ or $2m \le i' < j'<3m$, similar to that of Case 2, we also have that
       \begin{equation*}
         \mbox{rank}(A_{1,i} - A_{1,j})=N \Leftrightarrow c^{2m(u-v)+2(i'-j')\pm (a-b)}-1\ne 0
       \end{equation*}
for all $a, b=0,1$,
which holds from a similar analysis as  in Case 2.

\textbf{Case 4:} When $0 \le i' < m$ and $m \le j'<2m$, if $j'=i'+m$, then by \eqref{B3} we have
{\small       \begin{IEEEeqnarray*}{rCl}
 \hspace{-1mm} && \mbox{rank}(A_{1,i}-A_{1,j})\\
  \hspace{-1mm}&=& \mbox{rank}(z_{1,i}B'_{1,i'} - z_{1,j}B'_{1,j'})\\
   \hspace{-1mm}&=& \mbox{rank}(\left(\begin{array}{c}
               V_{i',0}(z_{1,i}B'_{1,i'} - z_{1,j}B'_{1,j'}) \\
               V_{i',1}(z_{1,i}B'_{1,i'} - z_{1,j}B'_{1,j'})
         \end{array}
\right))\\
 \hspace{-1mm}&=& \mbox{rank}(\left(\hspace{-1mm}\begin{array}{c}
 z_{1,i}\lambda_{i',0}V_{i',0}+z_{1,i}(\lambda_{i',0}\hspace{-1mm}-\hspace{-1mm}\lambda_{i',1})V_{i',1}- z_{1,j}\lambda_{j',0}V_{i',0} \\
 z_{1,i}\lambda_{i',1}V_{i',1}-z_{1,j}\lambda_{j',1}V_{i',1}- z_{1,j}(\lambda_{j',1}\hspace{-1mm}-\hspace{-1mm}\lambda_{j',0})V_{i',0}
         \end{array}
 \hspace{-2mm}\right))\\
 \hspace{-1mm}&=& \mbox{rank}(\left(\hspace{-1mm}\begin{array}{c}
 (z_{1,i}\lambda_{i',0}-z_{1,j}\lambda_{j',0})V_{i',0}+z_{1,i}(\lambda_{i',0}-\lambda_{i',1})V_{i',1}  \\
 (z_{1,i}\lambda_{i',1}-z_{1,j}\lambda_{j',1})V_{i',1}- z_{1,j}(\lambda_{j',1}-\lambda_{j',0})V_{i',0}
         \end{array}
\hspace{-1mm}\right))\\
 \hspace{-1mm}&=&\mbox{rank}(\left(\hspace{-1mm}\begin{array}{l}
 (z_{1,i}\lambda_{i',0}-z_{1,j}\lambda_{j',1})(V_{i',0}+V_{i',1})  \\
 (z_{1,i}\lambda_{i',1}-z_{1,j}\lambda_{j',1})V_{i',1}- z_{1,j}(\lambda_{j',1}-\lambda_{j',0})V_{i',0}
         \end{array}
\hspace{-1mm}\right))\\
 \hspace{-1mm}&=&N,
\end{IEEEeqnarray*}
}
which is equivalent to{\small
\begin{equation*}
  (z_{1,i}\lambda_{i',0}\hspace{-.5mm}-\hspace{-.5mm}z_{1,j}\lambda_{j',1})\left((z_{1,i}\lambda_{i',1}\hspace{-1mm}-\hspace{-1mm}z_{1,j}\lambda_{j',1})\hspace{-.5mm}
  +\hspace{-.5mm}z_{1,j}(\lambda_{j',1}\hspace{-1mm}-\hspace{-1mm}\lambda_{j',0})\right)\hspace{-.5mm}\ne\hspace{-.5mm} 0,
\end{equation*}}
 i.e.,
\begin{IEEEeqnarray*}{rCl}
\hspace{-2mm}&&\hspace{-1mm}(c^{2mu+2i'}-c^{2mv+2i'+1})\left((c^{2mu+2i'+1}-c^{2mv+2i'+1})\right.\\\hspace{-2mm}&& \hspace{45mm} \left.+c^{2mv}(c^{2i'+1}-c^{2i'})\right)\\
\hspace{-2mm}&=&\hspace{-1mm}
c^{2mv+2i'+1}(c^{2m(u-v)-1}-1)c^{2mv+2i'}(c^{2m(u-v)+1}-1)\\
\hspace{-2mm}&\ne& \hspace{-1mm}0.
\end{IEEEeqnarray*}
The above inequality always holds since $$0<|2m(u-v)\pm 1|\le 2m\lceil\frac{n}{n'}\rceil-2m+1<q-1;$$
Otherwise, similar to Case 2, we have that
\begin{equation*}
                \mbox{rank}(A_{1,i}-A_{1,j})=N
\end{equation*}
is equivalent to
$$c^{2m(u-v)+2(i'-j'+m)+a-b}-1\ne 0,~\mbox{for~all~}a, b=0,1,$$
which holds according to a similar analysis  as     in Case 2.

 \textbf{Case 5:} When $0 \le i' < m$ and $2m \le j'<3m$, if $j'=i'+2m$, then by \eqref{B3} we have
{\small
       \begin{IEEEeqnarray*}{rCl}
&& \mbox{rank}(A_{1,i}-A_{1,j})\\
&=&\mbox{rank}(z_{1,i}B'_{1,i'} - z_{1,j}B'_{1,j'})\\
&=&\mbox{rank}(\left(\begin{array}{c}
                  V_{i',0}(z_{1,i}B'_{1,i'} - z_{1,j}B'_{1,j'}) \\
                  V_{i',1}(z_{1,i}B'_{1,i'} - z_{1,j}B'_{1,j'})
                \end{array}
            \right))\\
&=&\mbox{rank}(\left(\hspace{-.5mm}
                \begin{array}{l}
 z_{1,i}\lambda_{i',0}V_{i',0}\hspace{-.5mm}+\hspace{-.5mm}z_{1,i}(\lambda_{i',0}\hspace{-1mm}-\hspace{-1mm}\lambda_{i',1})V_{i',1}\hspace{-.5mm} -\hspace{-.5mm} z_{1,j}\lambda_{j',0}V_{i',0} \\
                    z_{1,i}\lambda_{i',1}V_{i',1} - z_{1,j}\lambda_{j',1}V_{i',1}
                \end{array}\hspace{-.5mm}\right))\\
&=&\mbox{rank}(\left(\hspace{-.5mm}
                \begin{array}{l}
 (z_{1,i}\lambda_{i',0}\hspace{-.5mm}-\hspace{-.5mm} z_{1,j}\lambda_{j',0})V_{i',0}+z_{1,i}(\lambda_{i',0}\hspace{-.5mm}-\hspace{-.5mm}\lambda_{i',1})V_{i',1}  \\
                    (z_{1,i}\lambda_{i',1}- z_{1,j}\lambda_{j',1})V_{i',1}
 \end{array}\hspace{-.5mm}\right))\\
 &=&N\\
&\Leftrightarrow& c^{2mv+2i'+1}(c^{2m(u-v)-1}-1)c^{2mv+2i'}(c^{2m(u-v)+1}-1) \ne 0;
\end{IEEEeqnarray*}
}
which holds for a  similar reason as in Case 4;
Otherwise,
\begin{IEEEeqnarray*}{rCl}
\hspace{-3mm}     && \mbox{rank}(A_{1,i}-A_{1,j})\\
 \hspace{-3mm}&=&\mbox{rank}(z_{1,i}B'_{1,i'} - z_{1,j}B'_{1,j'})\\
 \hspace{-3mm}&=& \mbox{rank}
         (\left(\begin{array}{c}
           V_{i',j',0,0}(z_{1,i}B'_{1,i'} - z_{1,j}B'_{1,j'})\\
           V_{i',j',0,1}(z_{1,i}B'_{1,i'} - z_{1,j}B'_{1,j'})\\
           V_{i',j',1,0}(z_{1,i}B'_{1,i'} - z_{1,j}B'_{1,j'})\\
           V_{i',j',1,1}(z_{1,i}B'_{1,i'} - z_{1,j}B'_{1,j'})
         \end{array}\right))\\
\hspace{-3mm}&=&\mbox{rank}(\left(
                \begin{array}{l}
            z_{1,i}\lambda_{i',0}V_{i',j',0,0}- z_{1,j}\lambda_{j',0} V_{i',j',0,0}\\\hspace{15mm}+  z_{1,i}(\lambda_{i',0}-\lambda_{i',1}) V_{i',j',1,0}\\
                    z_{1,i}\lambda_{i',0}V_{i',j',0,1}- z_{1,j}\lambda_{j',1} V_{i',j',0,1}\\\hspace{15mm}+ z_{1,i}(\lambda_{i',0}-\lambda_{i',1}) V_{i',j',1,1} \\
                    z_{1,i}\lambda_{i',1}V_{i',j',1,0}- z_{1,j}\lambda_{j',0} V_{i',j',1,0}\\
                     z_{1,i}\lambda_{i',1}V_{i',j',1,1}- z_{1,j}\lambda_{j',1} V_{i',j',1,1}
                \end{array}\right))\\
\hspace{-3mm}&=&\mbox{rank}(\left(
                \begin{array}{l}
            (z_{1,i}\lambda_{i',0}- z_{1,j}\lambda_{j',0})V_{i',j',0,0}\\\hspace{18mm}+  z_{1,i}(\lambda_{i',0}-\lambda_{i',1}) V_{i',j',1,0}\\
                    (z_{1,i}\lambda_{i',0}-z_{1,j}\lambda_{j',1})V_{i',j',0,1}\\\hspace{18mm}+ z_{1,i}(\lambda_{i',0}-\lambda_{i',1}) V_{i',j',1,1}  \\
                    (z_{1,i}\lambda_{i',1}- z_{1,j}\lambda_{j',0}) V_{i',j',1,0}\\
                     (z_{1,i}\lambda_{i',1}- z_{1,j}\lambda_{j',1}) V_{i',j',1,1}
                \end{array}\right))\\
\hspace{-3mm}&=&N,
\end{IEEEeqnarray*}
which is equivalent to
\begin{equation*}
  z_{1,i}\lambda_{i',a}-z_{1,j}\lambda_{j',b}\ne 0,\mbox{~for~all~}a, b=0,1,
\end{equation*}
i.e.,
\begin{equation*}
c^{2mv+2(j'-2m)+1-b}(c^{2m(u-v)+2(i'-j'+2m)+a+b-1}-1)\ne 0
\end{equation*}
for all $a, b=0,1$,
which holds due to a similar analysis   as  in Case 2.

\textbf{Case 6:} When $m \le i' < 2m$ and $2m \le j'<3m$, similar to that of Case 5, if $j'=i'+m$, we  have
\begin{equation*}
         \mbox{rank}(A_{1,i} - A_{1,j})=N,
       \end{equation*}
is equivalent to
\begin{equation*}
  (c^{2m(u-v)+1}-1)(c^{2m(u-v)-1}-1)\ne 0,
\end{equation*}
      otherwise
       \begin{equation*}
         \mbox{rank}(A_{1,i} - A_{1,j})=N
       \end{equation*}
is equivalent to
\begin{equation*}
  c^{2m(u-v)+2(i'-j'+m)+a+b-1}-1\ne 0 \mbox{~for~all~}a, b=0,1.
\end{equation*}
        The above two inequalities always hold due to a similar reason as    in Case 5.

Combining the above 6 cases, we finish the proof.
     \end{proof}

Finally, we demonstrate to what extent  Theorem  \ref{Thm C4 result2} can reduce the field size $q$ of the $(n, k)$ MDS code $\mathcal{C}_4$. For example, when $n'=6$, $m=2$, and $n=24$. According to Theorem
\ref{Thm C4 result2}, we can set
\begin{equation*}
z_{0,i}=1,\,\, z_{1,i}=3^{4\lfloor\frac{i}{6}\rfloor},
\end{equation*}
for  $i \in [0,24)$ in \eqref{Eqn pc C4} and
\begin{IEEEeqnarray*}{rCl}
&&\lambda_{i',0}=\lambda_{i'+2,0}=\lambda_{i'+4,1}=3^{2i'},\\
&&\lambda_{i',1}=\lambda_{i'+2,1}=\lambda_{i'+4,0}=3^{2i'+1},
\end{IEEEeqnarray*}
in \eqref{Eqn pc C4'}  for $i'=0, 1$ over $\mathbf{F}_{17}$ with $3$ being the primitive element. Whereas, by Theorem \ref{Thm C4 result1}, the existence of the MDS code $\mathcal{C}_4$ requires a finite field with size larger than $92$.


\section{An $(n,k)$ MDS code $\mathcal{C}_5$ with the optimal update property and small sub-packetization over small finite fields} \label{sec:zigzag}

Note from \eqref{Eqn coding matrix C1} that the parity-check matrix of the MDS code $\mathcal{C}_1$ has a constraint, that is, block matrices $A_{t,i}$ should satisfy that $A_{t,j_1}A_{t,j_2}^{-1}$ is a scalar matrix over   $\mathbf{F}_q$  for all $j_1\equiv j_2(\bmod\, n')$ and $t\in [0,r)$,
which reduces the designing space for the parameters $\lambda_{i,0}, \ldots, \lambda_{i,r-1}$ in \eqref{Eqn_YB_code} to guarantee the  MDS property.
In this section, we propose another explicit $(n,k)$ MDS code which has a similar     structure as that of the  MDS code $\mathcal{C}_1$, but allows more flexible  choices of  $\lambda_{i,0}, \ldots, \lambda_{i,r-1}$, and thus can  further reduce the field size.

Let $N=r^{n'}$ and $n>n'$, where $n$ and $n'$ are two positive integer.
Construct an $(n,k)$   code $\mathcal{C}_5$ with longer code length given by \eqref{Eqn parity check eq} and \eqref{Eqn A power}, where $A_i$, $i\in [0,n)$ satisfy
\begin{equation}\label{Eqn Hadamard coding matrix}
\left(
                     \begin{array}{c}
                       V_{i,0} \\
                       V_{i,1} \\
                       \vdots \\
                       V_{i,r-1}
                     \end{array}
                   \right)A_i=\left(
                     \begin{array}{c}
                      \lambda_{i,0} V_{i,0} \\
                      \lambda_{i,1} V_{i,1} \\
                       \vdots \\
                      \lambda_{i,r-1} V_{i,r-1}
                     \end{array}
                   \right),
\end{equation}
with $\lambda_{i,t} \in \mathbf{F}_q\backslash\{0\}$ and $V_{i, t}$ being defined by \eqref{Eqn_Vt} and \eqref{B3} for $t\in [0, r)$.
The repair matrices and select matrices are respectively  defined by
\begin{equation}\label{Eqn Hadmard R}
R_{i,j}=\left\{
                      \begin{array}{ll}
                        V_{i,0}+V_{i,1}+\cdots+V_{i,r-1}, &\mbox{\ \ if\ \ } j\not\equiv i \bmod n', \\
                        I, & \mbox{\ \ otherwise.\ \ }
                      \end{array}
                    \right.
\end{equation}
and
\begin{equation}\label{Eqn Hadmard S}
S_{i,t}=V_{i,0}+V_{i,1}+\cdots+V_{i,r-1},~t\in[0,r).
\end{equation}

\begin{Theorem}\label{Thm repair}
Every failed node of the  code $\mathcal{C}_5$ can be regenerated by the repair matrices defined in \eqref{Eqn Hadmard R} and \eqref{Eqn Hadmard S} if $\lambda_{i,0},\lambda_{i,1},\cdots,\lambda_{i,r-1}$ are pairwise distinct for each $i\in [0,n)$. Furthermore, the repair bandwidth  for node $i$ ($i\in [0, n)$) is
\begin{equation*}
  \gamma_i = \left\{
                      \begin{array}{ll}
                      (1+\frac{(\lceil\frac{n}{n'}\rceil-1)(r-1)}{n-1})\gamma^*, &\mbox{\ \ if\ \ } 0\le i\% n'<n\%n', \\
                       (1+\frac{(\lfloor\frac{n}{n'}\rfloor-1)(r-1)}{n-1})\gamma^*, & \mbox{\ \ otherwise}.
                      \end{array}
                    \right.
\end{equation*}
\end{Theorem}
\begin{proof}
Firstly, for $i\in[0,n)$, by \eqref{Eqn Hadamard coding matrix}, we have
\begin{equation}\label{Eqn A^t}
\left(
                     \begin{array}{c}
                       V_{i,0} \\
                       V_{i,1} \\
                       \vdots \\
                       V_{i,r-1}
                     \end{array}
                   \right)A_i^t=\left(
                     \begin{array}{c}
                      \lambda_{i,0}^t V_{i,0} \\
                      \lambda_{i,1}^t V_{i,1} \\
                       \vdots \\
                      \lambda_{i,r-1}^t V_{i,r-1}
                     \end{array}
                   \right),~t\in[0,r).
\end{equation}
Then,
\begin{align*}
&\textrm{rank}(\left(
\begin{array}{c}
S_{i,0} A_{0,i}\\
S_{i,1} A_{1,i} \\
\vdots \\
S_{i,r-1} A_{r-1,i}
\end{array}
\right))\\=&\textrm{rank}(\left(
\begin{array}{c}
S_{i,0} \\
S_{i,1} A_{i} \\
\vdots \\
S_{i,r-1} A_{i}^{r-1}
\end{array}
\right))\\
=&\textrm{rank}(\left(
\begin{array}{c}
V_{i,0}+V_{i,1}+\cdots+V_{i,r-1} \\
\lambda_{i,0}V_{i,0}+\lambda_{i,1}V_{i,1}+\cdots+\lambda_{i,r-1}V_{i,r-1} \\
\vdots \\
\lambda_{i,0}^{r-1}V_{i,0}+\lambda_{i,1}^{r-1}V_{i,1}+\cdots+\lambda_{i,r-1}^{r-1}V_{i,r-1}
\end{array}
\right))\\=&\textrm{rank}(
\left(
  \begin{array}{cccc}
    1 & 1 & 1 & 1 \\
    \lambda_{i,0} & \lambda_{i,1} & \cdots & \lambda_{i,r-1} \\
    \vdots & \vdots & \ddots & \vdots \\
    \lambda_{i,0}^{r-1} & \lambda_{i,1}^{r-1} & \cdots & \lambda_{i,r-1}^{r-1}
  \end{array}
\right)\left(
                     \begin{array}{c}
                       V_{i,0} \\
                       V_{i,1} \\
                       \vdots \\
                       V_{i,r-1}
                     \end{array}
                   \right)).
\end{align*}
Obviously,  the rank is $N$  if $\lambda_{i,u}\ne \lambda_{i,v}$ for all $u, v\in[0,r)$ with $u\ne v$.

Next, we prove that \eqref{repair_node_requirement3 n-1} holds. By means of \eqref{Eqn_B3} and \eqref{Eqn A^t},
if $j\not\equiv i \bmod n'$, then we have
\begin{align*}
\textrm{rank} (\left(
\begin{array}{c}
R_{i,j} \\
S_{i,t} A_{t,j}
\end{array}
\right))
&=\textrm{rank}(\left(
\begin{array}{c}
 \sum\limits_{u=0}^{r-1}V_{i,u} \\
 \sum\limits_{u=0}^{r-1}V_{i,u}A_{j}^t
\end{array}
\right))\\
&=
\textrm{rank}(\left(
\begin{array}{c}
\sum\limits_{u=0}^{r-1}V_{i,u} \\
\sum\limits_{u=0}^{r-1}V_{i,j,u,0}A_{j}^t \\
\sum\limits_{u=0}^{r-1}V_{i,j,u,1}A_{j}^t \\
\vdots\\
\sum\limits_{u=0}^{r-1}V_{i,j,u,r-1}A_{j}^t
\end{array}
\right))\\
&=
\textrm{rank} (\left(
\begin{array}{c}
\sum\limits_{u=0}^{r-1}V_{i,u} \\
\lambda_{j,0}^t\sum\limits_{u=0}^{r-1}V_{i,j,u,0} \\
\lambda_{j,1}^t\sum\limits_{u=0}^{r-1}V_{i,j,u,1}\\
\vdots\\
\lambda_{j,r-1}^t\sum\limits_{u=0}^{r-1}V_{i,j,u,r-1}
\end{array}
\right))\\
&=\textrm{rank}(R_{i,j});
\end{align*}
Otherwise, we have
\begin{align*}
\textrm{rank} (\left(
\begin{array}{c}
R_{i,j} \\
S_{i,t} A_{t,j}
\end{array}
\right))&=
\textrm{rank} (\left(
\begin{array}{c}
 I \\
 \sum\limits_{u=0}^{r-1}V_{i,u}A_{j}^t
\end{array}
\right))\\&=
\textrm{rank} (\left(
\begin{array}{c}
 I \\
 \sum\limits_{u=0}^{r-1}\lambda_{j,u}^tV_{j,u}
\end{array}
\right))\\&=
\textrm{rank}(R_{i,j}).
\end{align*}
Therefore, by  \eqref{Eqn_RB} and \eqref{Eqn Hadmard R}, the repair bandwidth of node $i$ is
\begin{align*}
  \gamma_i\hspace{-1mm} &=\hspace{-1mm} \sum\limits_{j=0,j\ne i}^{n-1} \mbox{rank}(R_{i,j})\\
  \hspace{-1mm}&=\hspace{-1mm}(n-1)\frac{N}{r}\hspace{-1mm}+\hspace{-1mm}\frac{(r-1)N}{r}|\{j:j\in [0, n)\backslash\{i\}, j\equiv i \bmod n'\}|\\
  \hspace{-1mm}&=\hspace{-1mm}
  \left\{
                      \begin{array}{ll}
                      (1+\frac{(\lceil\frac{n}{n'}\rceil-1)(r-1)}{n-1})\gamma^*, &\mbox{\ \ if\ \ } 0\le i\% n'<n\%n', \\
                       (1+\frac{(\lfloor\frac{n}{n'}\rfloor-1)(r-1)}{n-1})\gamma^*, & \mbox{\ \ otherwise},
                      \end{array}
                    \right.
\end{align*}
where $\gamma^*=(n-1)\frac{N}{r}$ is the optimal value for repair bandwidth.
\end{proof}

\begin{Theorem}\label{Thm MDS}
The code $\mathcal{C}_5$  possesses the MDS property if
\begin{itemize}
  \item [(i)] $\lambda_{i,u}\ne \lambda_{j,v}$ for all $u,v\in[0,r)$ and $i,j\in[0,n)$ with $j\not\equiv i \bmod n'$,
  \item [(ii)] $\lambda_{i,u}\ne \lambda_{i+gn',u}$ for all $u\in[0,r)$, $g\in[1,\lceil \frac{n}{n'}\rceil)$, $i\in[0,n')$   with $i+gn'<n$.
\end{itemize}
\end{Theorem}
\begin{proof}
The proof can be proceeded in the same fashion as that of Theorem \ref{Thm MDS C1 explicit}.
\end{proof}

In the following, we give an assignment of the values $\lambda_{i,u}$, $i\in[0,n)$, $u\in [0,r)$ so that the requirements in Theorems \ref{Thm repair} and \ref{Thm MDS} can be satisfied.

\begin{Theorem}\label{Thm assignment}
The requirements in Theorems \ref{Thm repair} and \ref{Thm MDS} can be satisfied if   $q$ is a prime power such that
\begin{equation*}
  q >\left\{
                      \begin{array}{ll}
                       rn'(\lceil\frac{n}{rn'}\rceil-1)+(n\% n')r, & \mbox{\ \ if\ \ }0<n\% (rn')<n',\\
                       rn'\lceil\frac{n}{rn'}\rceil, &\mbox{\ \ otherwise}.
                      \end{array}
                    \right.
\end{equation*}
\end{Theorem}
\begin{proof}
If $0<n\% (rn')<n'$, then $\lceil\frac{n}{rn'}\rceil-1=\lfloor\frac{n}{rn'}\rfloor$ and $n\% (rn')=n\% n'$,
 let  $\xi_{i',v}^{(z)}$, $z\in[0,\lfloor\frac{n}{rn'}\rfloor)$,  $i'\in [0,n')$, $v\in[0,r)$, and $\xi_{i',v}^{(\lfloor\frac{n}{rn'}\rfloor)}$, $i'\in [0, n\% n')$, $v\in [0, r)$  be $rn'\lfloor\frac{n}{rn'}\rfloor+(n\%n')r$ pairwise distinct nonzero elements in $\mathbf{F}_q$; Otherwise, let  $\xi_{i',v}^{(z)}$, $z\in[0,\lceil\frac{n}{rn'}\rceil)$, $i'\in [0,n')$, $v\in[0,r)$ be $rn'\lceil\frac{n}{rn'}\rceil$ pairwise distinct nonzero elements in $\mathbf{F}_q$.  Then for $i=zrn'+un'+i'$, $i'\in [0,n')$, $u\in[0,r)$, $z\in[0,\lceil\frac{n}{rn'}\rceil)$, if we set
$\lambda_{i,t}=\xi_{i',t+u}^{(z)}$ for $i\in[0,n)$ and $t\in [0,r)$, where
the subscript $t+u$ is computed modulo $r$, it is easy to verify that the requirements in Theorems \ref{Thm repair} and \ref{Thm MDS} can be satisfied.
\end{proof}

In the following, we give a concrete example of the MDS code $\mathcal{C}_5$ according to Theorem \ref{Thm assignment}.

\begin{Example}\label{Ex2}
Let $r=2$, $n'=3$, and $n=12$, then the parity-check matrix of the $(12,10)$ MDS code $\mathcal{C}_5$ over $\mathbf{F}_{13}$ is defined through
\begin{equation*}
A_0=\left(
      \begin{array}{c}
        e_0 \\
        e_1 \\
        e_2 \\
        e_3 \\
        c e_4 \\
        c e_5 \\
        c e_6 \\
        c e_7
      \end{array}
    \right),~A_1=\left(
      \begin{array}{c}
        c^2 e_0 \\
        c^2 e_1 \\
        c^3 e_2 \\
        c^3 e_3 \\
        c^2 e_4 \\
        c^2 e_5 \\
        c^3 e_6 \\
        c^3 e_7
      \end{array}
    \right),~A_2=\left(
      \begin{array}{c}
        c^4 e_0 \\
        c^5 e_1 \\
        c^4 e_2 \\
        c^5 e_3 \\
        c^4 e_4 \\
        c^5 e_5 \\
        c^4 e_6 \\
        c^5 e_7
      \end{array}
    \right),
\end{equation*}
\begin{equation*}
A_3=\left(
      \begin{array}{c}
        c e_0 \\
        c e_1 \\
        c e_2 \\
        c e_3 \\
        e_4 \\
        e_5 \\
        e_6 \\
        e_7
      \end{array}
    \right),~A_4=\left(
      \begin{array}{c}
        c^3 e_0 \\
        c^3 e_1 \\
        c^2 e_2 \\
        c^2 e_3 \\
        c^3 e_4 \\
        c^3 e_5 \\
        c^2 e_6 \\
        c^2 e_7
      \end{array}
    \right),~A_5=\left(
      \begin{array}{c}
        c^5 e_0 \\
        c^4 e_1 \\
        c^5 e_2 \\
        c^4 e_3 \\
        c^5 e_4 \\
        c^4 e_5 \\
        c^5 e_6 \\
        c^4 e_7
      \end{array}
    \right),
\end{equation*}
\begin{equation*}
   A_6=\left(
      \begin{array}{c}
        c^6 e_0 \\
        c^6 e_1 \\
        c^6 e_2 \\
        c^6 e_3 \\
        c^7 e_4 \\
        c^7 e_5 \\
        c^7 e_6 \\
        c^7 e_7
      \end{array}
    \right),~A_7=\left(
      \begin{array}{c}
        c^8 e_0 \\
        c^8 e_1 \\
        c^9 e_2 \\
        c^9 e_3 \\
        c^8 e_4 \\
        c^8 e_5 \\
        c^9 e_6 \\
        c^9 e_7
      \end{array}
    \right),~A_{8}=\left(
      \begin{array}{c}
        c^{10} e_0 \\
        c^{11} e_1 \\
        c^{10} e_2 \\
        c^{11} e_3 \\
        c^{10} e_4 \\
        c^{11} e_5 \\
        c^{10} e_6 \\
        c^{11} e_7
      \end{array}
    \right),
\end{equation*}
\begin{equation*}
A_{9}=\left(
      \begin{array}{c}
        c^7 e_0 \\
        c^7 e_1 \\
        c^7 e_2 \\
        c^7 e_3 \\
        c^6 e_4 \\
        c^6 e_5 \\
        c^6 e_6 \\
        c^6 e_7
      \end{array}
    \right),~A_{10}=\left(
      \begin{array}{c}
        c^9 e_0 \\
        c^9 e_1 \\
        c^8 e_2 \\
        c^8 e_3 \\
        c^9 e_4 \\
        c^9 e_5 \\
        c^8 e_6 \\
        c^8 e_7
      \end{array}
    \right),~A_{11}=\left(
      \begin{array}{c}
        c^{11} e_0 \\
        c^{10} e_1 \\
        c^{11} e_2 \\
        c^{10} e_3 \\
        c^{11} e_4 \\
        c^{10} e_5 \\
        c^{11} e_6 \\
        c^{10} e_7
      \end{array}
    \right),
\end{equation*}
where $c=2$.
\end{Example}

Similar to the MDS code $\mathcal{C}_1$, we have the following result.

\begin{Theorem}\label{C5-update}
The MDS code   $\mathcal{C}_5$ has the optimal update property.
\end{Theorem}

\section{Comparisons}
In this section, we give  comparisons  of some key parameters among the  proposed MDS codes  and some existing  notable  MDS  codes.

\begin{table*}[htbp]
\begin{center}
\caption{A comparison of some key parameters among the $(n,k)$ MDS codes proposed in this paper and some existing  notable $(n,k)$ MDS  codes, where we set $n=sn'$ for convenience and $r=n-k$}\label{Table comp}
\setlength{\tabcolsep}{3pt}
\begin{tabular}{|c|c|c|c|c|c|}
\hline
&Sub-packetization& \multirow{2}{*}{Field size } &The ratio of repair bandwidth  & \multirow{2}{*}{Remark}& \multirow{2}{*}{References} \\
&level $N$  &  &to the optimal value $\gamma^*$  &  &\\
\hline\hline
The new MDS code $\mathcal{C}_1$  & $r^{n'}$ & $q>rn'\lceil\frac{n}{rn'}\rceil$, $r|(q-1)$ & $=1+\frac{(s-1)(r-1)}{n-1}<1+\frac{r}{n'}$ & Optimal update & Thms \ref{Thm C1}-\ref{C1-update}\\
\hline
The new MDS code $\mathcal{C}_5$  & $r^{n'}$ &  $q>rn'\lceil\frac{n}{rn'}\rceil$ & $=1+\frac{(s-1)(r-1)}{n-1}<1+\frac{r}{n'}$ & Optimal update & Thms \ref{Thm assignment}, \ref{C5-update}\\
\hline
The RTGE code 2 &   $O(r^{r\tau}\log n)$  & $O(n)$    & $\le 1+\frac{1}{\tau}$  & $\tau>0$&  \cite{Rawat}\\
\hline
The YB code 1  &$r^{n}$   &  $q\ge rn$  &  $1$ ~~(Optimal)&   Optimal update   & \cite{Barg1}\\
\hline\hline
The new MDS code $\mathcal{C}_2$  & $r^{n'-1}$ &  $q>r\lceil{n'\over r}\rceil(\lceil{n\over n'}\rceil-1)+n'$ & $=1+\frac{(s-1)(r-1)}{n-1}<1+\frac{r}{n'}$ &  & Thm \ref{Thm MDS C2}\\
\hline
The new MDS code $\mathcal{C}_3$  & $r^{n'-1}$ & $
                      \begin{array}{ll}
                      q> s, q \mbox{~is~odd},&\hspace{-2mm} \mbox{if~} r \mbox{~is~even}\\
                       q>sr, &\hspace{-2mm}\mbox{otherwise}
                      \end{array}$ & $=1+\frac{(s-1)(r-1)}{n-1}<1+\frac{r}{n'}$ &  & Thm \ref{Thm MDS C3}\\
\hline
The improved YB code 2    &$r^{n-1}$   &  $q>r$  &  $1$ ~~(Optimal)&       & \cite{YiLiu}\\
\hline
 Shortened duplication-zigzag  &    $r^{n'-1}$  &  $q>s$    &  $=1+\frac{(s-1)(r-1)}{n-1}<1+\frac{r}{n'}$ & &  \cite{Zigzag} \\
\hline\hline
 The new MDS code $\mathcal{C}_4$ & $r^{n'\over {r+1}}$ &  $\begin{array}{ll}q> \frac{2n}{3}, & \mbox{if~} r=2\\ q>N{n-1\choose r-1}+1, & \mbox{if~} r>2 \end{array}$ & $=1+\frac{(s-1)(r-1)}{n-1}<1+\frac{r}{n'}$ & Implicit when $r>2$ & Thms \ref{Thm C4 result1}, \ref{Thm C4 result2} \\
\hline
The RTGE code 1 &   $r^{\tau}$  & $q>n^{(r-1)N+1}$    & $=1+\frac{(s-1)(r-1)}{n-1}< 1+\frac{1}{\tau}$  & $\begin{array}{l}\tau \mbox{~is~an~integer}\\1\le \tau \le \lceil\frac{n}{r} \rceil-1  \end{array}$ &  \cite{Rawat}\\
\hline
 Long code $\mathcal{C}_4'$ &    $r^{\frac{n}{r+1}}$  &  $q>N{n-1\choose r-1}+1$    &  $1$ ~~(Optimal) & Implicit when $r>2$  & Thms \ref{Thm_C4'_MDS}, \ref{Thm_C4'_RB}  \\
\hline
\end{tabular}
\end{center}
\end{table*}

\begin{table*}[htbp]
\begin{center}
\caption{A comparison of some key parameters among the  MDS codes $\mathcal{C}_4$ and the RTGE code 1 under some specific code lengths for $r=2$}\label{Table comp2}
\begin{tabular}{|c|c|c|c|c|c|}
\hline
& Code length & Number of& Sub-packetization&  Field size   &  \multirow{2}{*}{$\frac{\gamma}{\gamma^*}$}  \\
& $n$ &  parties $r$ & level $N$ & $q$ &    \\
\hline\hline
\multirow{3}{*}{The new MDS code $\mathcal{C}_4$}  & $12$ &  $2$ & $2^2$ & $3^2$ & $1+\frac{1}{11}$  \\
 & $18$ &  $2$ & $2^{2}$ & $13$ & $1+\frac{2}{17}$  \\
  & $24$ &  $2$ & $2^{2}$ & $17$ & $1+\frac{3}{23}$ \\
\hline
\multirow{3}{*}{The RTGE code 1}  & $12$ &  $2$ & $2^3$ & $>10^{9}$ & $1+\frac{1}{11}$ \\
 & $18$ &  $2$ & $2^{3}$ & $>10^{11}$ & $1+\frac{2}{17}$ \\
  & $24$ &  $2$ & $2^{3}$ & $>10^{12}$ & $1+\frac{3}{23}$  \\
\hline
\end{tabular}
\end{center}
\end{table*}

\begin{table*}[htbp]
\begin{center}
\caption{A comparison of some key parameters among the  MDS codes $\mathcal{C}_4$ and the RTGE code 1 under some specific code lengths for $r=3$}\label{Table comp3}
\setlength{\tabcolsep}{3.5pt}
\begin{tabular}{|c|c|c|c|c|c||c|}
\hline
& Code length & Number of& Sub-packetization&  Field size   &  \multirow{2}{*}{$\frac{\gamma}{\gamma^*}$}  &  \multirow{2}{*}{Remark} \\
& $n$ &  parties $r$ & level $N$ & $q$ &   &  \\
\hline\hline
\multirow{2}{*}{The new MDS code $\mathcal{C}_4$}  & $24$ &  $3$ & $3^3$ & $>6831$ & $1+\frac{1}{23}$  &\multirow{2}{*}{Implicit construction} \\
 & $36$ &  $3$ & $3^{3}$ & $>16065$ & $1+\frac{2}{35}$ & \\
\hline
\multirow{2}{*}{The RTGE code 1}  & $24$ &  $3$ & $3^4$ & $>10^{224}$ & $1+\frac{1}{23}$ &\\
& $36$ &  $3$ & $3^{4}$ & $>10^{253}$ & $1+\frac{2}{35}$  & \\
\hline
\end{tabular}
\end{center}
\end{table*}

\begin{table*}[htbp]
\begin{center}
\caption{A comparison of some key parameters among the  MDS codes $\mathcal{C}_4$ and the RTGE code 1 under some specific code lengths for $r=4$}\label{Table comp4}
\setlength{\tabcolsep}{3.5pt}
\begin{tabular}{|c|c|c|c|c|c||c|}
\hline
& Code length & Number of& Sub-packetization&  Field size   &  \multirow{2}{*}{$\frac{\gamma}{\gamma^*}$}  & \multirow{2}{*}{Remark}\\
& $n$ &  parties $r$ & level $N$ & $q$ &   &  \\
\hline\hline
\multirow{2}{*}{The new MDS code $\mathcal{C}_4$}  & $40$ &  $4$ & $4^4$ & $>2339584 $ & $1+\frac{1}{13}$ &  \multirow{2}{*}{Implicit construction}\\
 & $60$ &  $4$ & $4^{4}$ & $>8322304$ & $1+\frac{6}{59}$ & \\
\hline
\multirow{2}{*}{The RTGE code 1}  & $40$ &  $4$ & $4^5$ & $>10^{4923}$ & $1+\frac{1}{13}$  &\\
& $60$ &  $4$ & $4^{5}$ & $>10^{5464}$ & $1+\frac{6}{59}$  &\\
\hline
\end{tabular}
\end{center}
\end{table*}

Table \ref{Table comp} compares the details of these codes, while Tables \ref{Table comp2}-\ref{Table comp4} compare the new MDS code $\mathcal{C}_4$  and the RTGE code 1 in terms of the sub-packetization level,  the field size, and the repair bandwidth for $r=2,3$ and $4$, respectively.
From these tables, we see that the proposed MDS codes have the following advantages:
\begin{itemize}

\item  The new MDS codes $\mathcal{C}_1$, $\mathcal{C}_2$, $\mathcal{C}_3$, and $\mathcal{C}_5$  can support any number of parity nodes while the shortened duplication-zigzag code\footnote{Note that the code length of the  duplication-zigzag code in \cite{Zigzag}   is  in the form of $uk'+2$ with $uk'\gg2$, in order to  do a fair comparison under the same code length,   we delete two nodes of the duplication-zigzag code in \cite{Zigzag}  and term the resultant code as  shortened duplication-zigzag code.} in \cite{Zigzag} can only support two parity nodes.

\item  The new MDS codes $\mathcal{C}_1$ and $\mathcal{C}_5$ have the optimal update property.

\item The new $(n=sn',k)$ MDS codes derived in this paper indeed have a small sub-packetization level $N$. Specifically, $N=r^{n'}$ for the codes $\mathcal{C}_1$ and $\mathcal{C}_5$, $N=r^{\frac{n'}{r+1}}$ for the code $\mathcal{C}_4$, and  $N=r^{n'-1}$ for the codes $\mathcal{C}_2$ and $\mathcal{C}_3$. Note that  $n'$ can be fixed as a constant. Consequently, for each new MDS code, the sub-packetization level  can be  a constant, which is independent of  code length $n$.

\item Compared with   the RTGE code 1 in  \cite{Rawat}, when $n'=r\tau$, the new explicit MDS codes $\mathcal{C}_1$, $\mathcal{C}_2$, $\mathcal{C}_3$, and $\mathcal{C}_5$ are built on  much smaller finite  fields,  but have larger sub-packetization levels. Besides,  all the proposed MDS codes have   the same repair bandwidth as the RTGE code 1 in  \cite{Rawat} under the same parameters $n$ and $k$.

\item Particularly, the new $(n,k)$ MDS code $\mathcal{C}_4$   has  not only a smaller sub-packetization level,  but also  a much  smaller finite field when compared to the RTGE code 1.

Nevertheless, the code $\mathcal{C}_4$  is explicit  only for $r=2$, which requires a finite field with size $q>\frac{2n'}{3}\lceil\frac{n}{n'}\rceil$. For $r>2$,   further investigation is needed to find the explicit construction.

\item  In contrast to  RTGE code 2 in  \cite{Rawat}, which has sub-packetization growing logarithmically with the code length $n$, the new codes have   smaller  sub-packetizations. For example, the sub-packetization level of the MDS code $\mathcal{C}_5$ is around $\frac{1}{\log n}$  times that of the RTGE code 2 in  \cite{Rawat}   when $n'=r\tau$.

\item The  RTGE codes 1 and 2 in \cite{Rawat}  show that it is possible to trade repair bandwidth for sub-packetization, while the  proposed codes $\mathcal{C}_1$, $\mathcal{C}_2$, $\mathcal{C}_3$, and $\mathcal{C}_5$ further  show that it is possible to trade sub-packetization for field size base on the RTGE code 1, as these new codes are explicit and are over small finite fields.
\end{itemize}
In addition to the above advantages,  the new codes $\mathcal{C}_1$-$\mathcal{C}_5$ have a defect that they do not possess the load balancing property   as some of the helper nodes contribute a higher
amount of data during the node repair process. Whereas, the RTGE code 2 in  \cite{Rawat} is load balanced, where all the contacted nodes
provide (approximately) the same amount of information during
the repair process.

\section{Conclusion}
In this paper, we provided  a  powerful transformation that can greatly reduce the  sub-packetization level $N$ of the original codes with respect to the same code length $n$. Four applications of the transformation were demonstrated, three of which  are explicit and over a small finite field. In addition, another  explicit MDS code construction over a small  finite field and with  small sub-packetization level, small repair bandwidth as well as  the optimal update property was presented. The comparisons show that the obtained MDS codes outperform  existing MDS codes in terms of the field size and/or the sub-packetization level.
Extending our transformation and constructions to the case of $d<n-1$ or  multiple node failures are part of our ongoing work.

\appendices

\section{Proof of Theorem \ref{Thm general MDS}}\label{sec:Appen1}

Before proving Theorem \ref{Thm general MDS}, let us introduce some necessary definitions and results on determinants.

\begin{Definition}[\cite{mirsky2012introduction}]\label{Def minor}
A $k$-rowed minor of an $n$-rowed determinant $D=det(a_{i,j})_{i\in[0, n), j\in [0, n)}$ is any  $k$-rowed determinant obtained when $n-k$ rows and $n-k$ columns are deleted from $D$. The $k$-rowed minor obtained from $D$ by retaining only the elements belonging to rows $r_0,\ldots,r_{k-1}$ and columns $s_0,\ldots,s_{k-1}$ will be denoted by
\begin{equation*}
D(r_0,\ldots,r_{k-1}|s_0,\ldots,s_{k-1}).
\end{equation*}
The cofactor $\widetilde{D}(r_0,\ldots,r_{k-1}|s_0,\ldots,s_{k-1})$ of the minor $D(r_0,\ldots,r_{k-1}|s_0,\ldots,s_{k-1})$ in a determinant  $D$ is defined as
\begin{align*}
&\widetilde{D}(r_0,\ldots,r_{k-1}|s_0,\ldots,s_{k-1})\\=&(-1)^{r_0+\ldots+r_{k-1}+s_0+\ldots+s_{k-1}}D(r_{k},\ldots,r_{n-1}|s_{k},\ldots,s_{n-1}),
\end{align*}
where $r_{k},\ldots,r_{n-1}$ are the $n-k$ numbers among $0,\ldots, n-1$ other than $r_0,\ldots,r_{k-1}$ and $s_{k},\ldots,s_{n-1}$ are the $n-k$ numbers among $0,\ldots, n-1$ other than $s_0,\ldots,s_{k-1}$.
\end{Definition}

\begin{Lemma}[Laplace's expansion theorem \cite{mirsky2012introduction}]\label{Lemma Laplace}
Let $D$ be an $n$-rowed determinant, and let $r_0, \ldots, r_{k-1}$ be integers such that $0\le k<n-1$ and $0\le r_0<\ldots<r_{k-1}<n$. Then
 \begin{equation*}
 \begin{split}
 D=&\sum_{0\le u_0< \ldots< u_{k-1}<n}D(r_0, \ldots, r_{k-1}|u_0, \ldots, u_{k-1}) \\
    & \hspace{22mm}\times\widetilde{D}(r_0, \ldots, r_{k-1}|u_0, \ldots, u_{k-1}).
\end{split}
 \end{equation*}
\end{Lemma}

\begin{Proposition}\label{Pro1}
Let $u\ge 2$  and let
{\small
\begin{equation*}
B=\hspace{-.5mm}\left(\hspace{-1.5mm}
\begin{array}{cccc}
y_{0,0}B_{0,0} & y_{0,1}B_{0,1} & \cdots &  y_{0,u-1}B_{0,u-1} \\
y_{1,0}B_{1,0} & y_{1,1}B_{1,1} & \cdots &  y_{1,u-1}B_{1,u-1} \\
\vdots & \vdots & \ddots &   \vdots \\
y_{u-1,0}B_{u-1,0} & y_{u-1,1}B_{u-1,1} & \cdots &  y_{u-1,u-1}B_{u-1,u-1} \end{array}
\hspace{-2mm}\right)
\end{equation*}
}be a block matrix of order $uN$ over a certain finite field $\mathbf{F}_q$, where $y_{i,j}$ is an indeterminate in  $\mathbf{F}_q$ and $B_{i,j}$ is a full rank matrix of order $N$  for  $i,j\in[0,u)$. Then $\det (B)$ is a homogeneous polynomial of degree $uN$   which includes the term
\begin{equation}\label{Eqn det term}
\left( \prod\limits_{t=0}^{u-1}\det(B_{t,t}) \right)y_{0,0}^Ny_{1,1}^N\cdots y_{u-1,{u-1}}^N.
\end{equation}
\end{Proposition}
\begin{proof}
Clearly, $\det(B)$ is a $uN$-rowed determinant, the expansion of which includes $(uN)!$ terms, where each term is a monomial of degree $uN$. Therefore,   $\det (B)$ is a homogeneous polynomial of degree $uN$.  In the following, we prove that $\det(B)$ includes the term in \eqref{Eqn det term} by induction.

Let $D=\det(B)$, when $u=2$, then by Definition \ref{Def minor} and Lemma \ref{Lemma Laplace}, we can get \eqref{Eqn_pro1_D_1} in the next page,
\begin{figure*}[hb]
\hrulefill
 \begin{align}
\nonumber D&=D(0,\ldots,N-1|0,\ldots,N-1)\widetilde{D}(0,\ldots,N-1|0,\ldots,N-1)\\
\nonumber&\hspace{10mm}+\sum_{0\le j_0<\ldots<j_{N-1}<2N\atop (j_0,\ldots,j_{N-1})\ne(0,\ldots,N-1)} D(0,\ldots,N-1|j_0, \ldots, j_{N-1})\widetilde{D}(0,\ldots,N-1|j_0, \ldots, j_{N-1})\\
\nonumber&=\det(y_{0,0}B_{0,0})\det(y_{1,1}B_{1,1}) +\sum_{0\le j_0< \ldots<j_{N-1}<2N\atop (j_0,\ldots,j_{N-1})\ne(0,\ldots,N-1)} D(0,\ldots,N-1|j_0, \ldots, j_{N-1})\widetilde{D}(0,\ldots,N-1|j_0, \ldots, j_{N-1})\\
\label{Eqn_pro1_D_1}&=\left( \prod\limits_{t=0}^{1}\det(B_{t,t}) \right)y_{0,0}^Ny_{1,1}^N +\sum_{0\le j_0<\ldots<j_{N-1}<2N\atop (j_0,\ldots,j_{N-1})\ne(0,\ldots,N-1)} D(0,\ldots,N-1|j_0, \ldots, j_{N-1})\widetilde{D}(0,\ldots,N-1|j_0, \ldots, j_{N-1}).
\end{align}
\end{figure*}
which implies that $D$ includes the term in \eqref{Eqn det term}.

Assume that the induction hypothesis holds, i.e., $D$ includes the term in \eqref{Eqn det term} for  $u=v\ge 2$. Then, when $u=v+1$, similarly, we can obtain \eqref{Eqn_pro1_D_2} in the next page.
\begin{figure*}[hb]
\hrulefill
\begin{align}
\nonumber D=&D(vN,\ldots,(v+1)N-1|vN,\ldots,(v+1)N-1)\widetilde{D}(vN,\ldots,(v+1)N-1|vN,\ldots,(v+1)N-1)\\
\nonumber &+\sum_{0\le j_0< \ldots<j_{N-1}<(v+1)N\atop (j_0,\ldots,j_{N-1})\ne(vN,\ldots,(v+1)N-1)} D(vN,\ldots,(v+1)N-1|j_0, \ldots, j_{N-1})\widetilde{D}(vN,\ldots,(v+1)N-1|j_0, \ldots, j_{N-1})\\
\label{Eqn_pro1_D_2}\begin{split}
 =&\det(y_{v,v}B_{v,v})\widetilde{D}(vN,\ldots,(v+1)N-1|vN,\ldots,(v+1)N-1)\\
   &+\sum_{0\le j_0< \ldots< j_{N-1}<(v+1)N\atop (j_0,\ldots,j_{N-1})\ne(vN,\ldots,(v+1)N-1)} D(vN,\ldots,(v+1)N-1|j_0, \ldots, j_{N-1})\widetilde{D}(vN,\ldots,(v+1)N-1|j_0, \ldots, j_{N-1}).
\end{split}
\end{align}
\end{figure*}
Note from Definition \ref{Def minor} that $\widetilde{D}(vN,\ldots,(v+1)N-1|vN,\ldots,(v+1)N-1)$ is a $vN$-rowed determinant, which
includes the term $$\left( \prod\limits_{t=0}^{v-1}\det(B_{t,t}) \right)y_{0,0}^Ny_{1,1}^N\cdots y_{v-1,{v-1}}^N$$ by the induction hypothesis. Hence, $D$ includes the term $$\left( \prod\limits_{t=0}^{v}\det(B_{t,t}) \right)y_{0,0}^Ny_{1,1}^N\cdots y_{v,{v}}^N.$$

Based on the above analysis, we proved that $\det(B)$ includes the term in  \eqref{Eqn det term} for any $u\ge 2$.
\end{proof}

\textit{\textbf{Proof of Theorem \ref{Thm general MDS}:}}
By \eqref{Eqn general coding matrix},  the parity-check matrix of the new $(n, k)$ code is
\begin{equation*}
A=\left(
\begin{array}{cccc}
A_{0,0} & A_{0,1} & \cdots & A_{0,n-1} \\
A_{1,0} & A_{1,1} & \cdots & A_{1,n-1} \\
\vdots & \vdots & \ddots & \vdots \\
A_{r-1,0} & A_{r-1,1} & \cdots  & A_{r-1,n-1}
\end{array}
\right)
\end{equation*}
with the $j$-th block column being
\begin{equation*}
  \left(
    \begin{array}{c}
      x_{0,j}A'_{0,j\%n'}\\
      x_{1,j}A'_{1,j\%n'} \\
      \vdots \\
      x_{r-1,j}A'_{r-1,j\%n'} \\
    \end{array}
  \right).
\end{equation*}
Then the new code is MDS if and only if any $r\times r$ sub-block matrix of $A$ is nonsingular.

For any $J=\{j_0,j_1,\cdots,j_{r-1}\}\subset [0,n)$, let $P_J$ be  the $r\times r$ sub-block matrix of $A$ formed by the $r$ block columns indicated by $J$, i.e.,
\begin{equation*}
P_J= \hspace{-1mm}\left(\hspace{-1mm}\begin{array}{ccc}
        x_{0,j_0}A'_{0,j_0\%n'}  & \cdots & x_{0,j_{r-1}}A'_{0,j_{r-1}\%n'} \\
        x_{1,j_0}A'_{1,j_0\%n'}  & \cdots &x_{1,j_{r-1}} A'_{1,j_{r-1}\%n'} \\
        \vdots  & \ddots & \vdots \\
      x_{r-1,j_0} A'_{r-1,j_0\%n'}  & \cdots &x_{r-1,j_{r-1}} A'_{r-1,j_{r-1}\%n'}
      \end{array}
\hspace{-1mm} \right),
\end{equation*}
which is nonsingular if $\det(P_J)$ is nonzero. Define $P=\prod\limits_{J\subset [0,n), |J|=r}P_J$, then $\det(P)=\prod\limits_{J\subset [0,n),|J|=r}\det(P_J)$. Thus, it suffices to prove that there is an assignment to the variables $x_{i,j}$, $i\in[0,r)$, $j\in[0,n)$  that does not evaluate $\det(P)$ to zero.

By Proposition \ref{Pro1}, $\det(P_J)$ is a homogeneous polynomial of degree $rN$ which includes the term
\begin{equation*}
 \left(\prod\limits_{t=0}^{r-1}\det(A'_{t,j_t\%n'}) \right)x_{0,j_0}^Nx_{1,j_1}^N\cdots x_{r-1,j_{r-1}}^N.
\end{equation*}
Then, $\det(P)$ is  a homogeneous polynomial  of degree $rN{n\choose r}$, where
each indeterminate $x_{i,j}$  has degree at most $N{n-1\choose r-1}$.
 Therefore, by Lemma \ref{Lemma Comb Null}, if $q>N{n-1\choose r-1}+1$, then there are  $x_{0,0}, \ldots, x_{0, n-1}, \ldots, x_{r-1,0}, \ldots, x_{r-1,n-1}\in \mathbf{F}_q\backslash\{0\}$  that does not evaluate $\det(P)$ to zero.  This finishes the proof. \hfill$\square$

\section{Proof of Theorem \ref{Thm_C4'_RB}}\label{sec:Appen2}

 The new storage code $\mathcal{C}'_4$   has the optimal repair bandwidth if and only if \eqref{repair_node_requirement1 n-1} and \eqref{repair_node_requirement3 n-1}  hold.

(i) Firstly, by \eqref{Eqn A=yB}, \eqref{Eqn pc C4'}, and \eqref{eqn the definition of repair matrix of base long MDS code}, we determine the necessary and sufficient conditions for \eqref{repair_node_requirement1 n-1} according to the following two cases.
\begin{itemize}
    \item []Case 1: For any $i'\in [0,rm)$, let $u=\lfloor \frac{i'}{m}\rfloor$, then we have
\begin{align*}
              & \mbox{rank}(\left(
                \begin{array}{c}
                    S'_{i',0}A'_{0,i'} \\
                    S'_{i',1}A'_{1,i'} \\
                    \vdots \\
                    S'_{i',r-1}A'_{r-1,i'}
 \end{array}\right))\\=&\mbox{rank}(\left(
                \begin{array}{c}
                   y_{0,i'} V_{i',u}B'_{0,i'} \\
                   y_{1,i'} V_{i',u}B'_{1,i'} \\
                    \vdots \\
                   y_{r-1,i'} V_{i',u}B'_{r-1,i'}
                \end{array}\right))\\ =&\mbox{rank}(\left(\hspace{-1mm}
                \begin{array}{c}
                    V_{i',u} \\
 \lambda_{i',u}V_{i',u}+\sum\limits_{a=0,a\ne u}^{r-1}(\lambda_{i',u}-\lambda_{i',a})V_{i',a} \\
                    \vdots \\
 \lambda_{i',u}^{r-1}V_{i',u}+\sum\limits_{a=0,a\ne u}^{r-1}(\lambda_{i',u}^{r-1}-\lambda_{i',a}^{r-1})V_{i',a}
                \end{array}\hspace{-1mm}\right)),
\end{align*}
which is of full rank if and only if \eqref{Eqn_p_fThm10_1} in the next page holds,
\begin{figure*}[hb]
\hrulefill
\begin{align}
\nonumber & \left|
                \begin{array}{ccccccc}
                    0&\cdots&0 & 1 & 0 & \cdots & 0 \\
                    \lambda_{i',u}-\lambda_{i',0} &\cdots&\lambda_{i',u}-\lambda_{i',u-1} & \lambda_{i',u} & \lambda_{i',u}-\lambda_{i',u+1} &\cdots&\lambda_{i',u}-\lambda_{i',r-1}\\
                    \vdots&\ddots&\vdots & \vdots&\vdots&\ddots & \vdots  \\
                    \lambda_{i',u}^{r-1}-\lambda_{i',0}^{r-1} &\cdots&\lambda_{i',u}^{r-1}-\lambda_{i',u-1}^{r-1} & \lambda_{i',u}^{r-1} & \lambda_{i',u}^{r-1}-\lambda_{i',u+1}^{r-1} &\cdots&\lambda_{i',u}^{r-1}-\lambda_{i',r-1}^{r-1}
                \end{array}\right|\\
\nonumber                 =&(-1)^{r-1}\left|
                \begin{array}{ccccccc}
                    1&\cdots&1 & 1 & 1 & \cdots & 1 \\
                    \lambda_{i',0} &\cdots&\lambda_{i',u-1} & \lambda_{i',u} & \lambda_{i',u+1} &\cdots&\lambda_{i',r-1}\\
                    \vdots&\ddots&\vdots & \vdots&\vdots&\ddots & \vdots  \\
        \lambda_{i',0}^{r-1} &\cdots&\lambda_{i',u-1}^{r-1} & \lambda_{i',u}^{r-1} & \lambda_{i',u+1}^{r-1} &\cdots&\lambda_{i',r-1}^{r-1}
                \end{array}\right|\\
\label{Eqn_p_fThm10_1}                \ne &0,
\end{align}
\end{figure*}
i.e., $\lambda_{i',0},\lambda_{i',1},\cdots,\lambda_{i',r-1}$ are pairwise distinct.

    \item []Case 2: For $i'\in [rm,(r+1)m)$,
        \begin{IEEEeqnarray*}{rCl}
              &&\mbox{rank}(\left(
                \begin{array}{c}
                    S'_{i',0}A'_{0,i'} \\
                    S'_{i',1}A'_{1,i} \\
                    \vdots \\
                    S'_{i',r-1}A'_{r-1,i'}
 \end{array}\right))\\&=&\mbox{rank}(\left(
                \begin{array}{c}
 y_{0,i'}(V_{i',0}+\cdots+V_{i',r-1})B'_{0,i'} \\
 y_{1,i'}(V_{i',0}+\cdots+V_{i',r-1})B'_{1,i'} \\
                    \vdots \\
                    y_{r-1,i'}( V_{i',0}+\cdots+V_{i',r-1} )B'_{r-1,i'}
                \end{array}\right)) \\ &=&\mbox{rank}(\left(
                \begin{array}{c}
                    V_{i',0}+\cdots+V_{i',r-1} \\
                    \lambda_{i',0}V_{i,0}+\cdots+\lambda_{i',r-1}V_{i,r-1} \\
                    \vdots \\
                    \lambda_{i',0}^{r-1}V_{i',0}+\cdots+\lambda_{i',r-1}^{r-1}V_{i',r-1}
                \end{array}\right))\\&=&N\\  &\Leftrightarrow& \left|
                \begin{array}{ccc}
                    1&\cdots&1 \\
 \lambda_{i',0}&\cdots&\lambda_{i',r-1}\\
                    \vdots&\ddots&\vdots  \\
 \lambda^{r-1}_{i',0}&\cdots&\lambda^{r-1}_{i',r-1}
                \end{array}\right|\ne0,
        \end{IEEEeqnarray*}
which holds if and only if $\lambda_{i',0},\lambda_{i',1},\cdots,\lambda_{i',r-1}$ are pairwise distinct.
\end{itemize}

(ii) Secondly, by  \eqref{B3}, \eqref{Eqn A=yB}, \eqref{Eqn reint A}, and \eqref{eqn the definition of repair matrix of base long MDS code}, we establish the necessary and sufficient conditions for \eqref{repair_node_requirement3 n-1} according to the following four cases.
\begin{itemize}
    \item []Case 1: For $t\in [0,r)$ and $i',j'\in[0,rm)$ with  $i'\ne j'$, let $u=\lfloor \frac{i'}{m}\rfloor$ and $v=\lfloor \frac{j'}{m}\rfloor$.
            If $j'\not\equiv i'$ mod $m$, then  we have
            \begin{IEEEeqnarray*}{rCl}
  &&\mbox{rank}(\left(
                    \begin{array}{c}
                        R'_{i',j'}\\
                        S'_{i',t}A'_{t,j'}
                    \end{array}\right))\\&=& \mbox{rank}(\left(
                    \begin{array}{c}
                        V_{i',u}\\
                        V_{i',u}A'_{t,j'}
                    \end{array}\right))\\&=& \mbox{rank}(\left(
                    \begin{array}{c}
                        V_{i',u}\\
                        V_{i',j',u,0}B'_{t,j'} \\
                        \vdots\\
                         V_{i',j',u,v-1}B'_{t,j'} \\
                         V_{i',j',u,v}B'_{t,j'} \\
                         V_{i',j',u,v+1}B'_{t,j'} \\
                        \vdots\\
                          V_{i',j',u,r-1}B'_{t,j'}
 \end{array}\right))\\&=&\mbox{rank}(\left(\hspace{-2mm}
                    \begin{array}{c}
                        V_{i',u}\\
                       \lambda^{t}_{j',0} V_{i',j',u,0} \\
                       \vdots\\
                       \lambda^{t}_{j',v-1} V_{i',j',u,v-1} \\
                       \lambda^{t}_{j',v} V_{i',j',u,v}\hspace{-1mm}+\hspace{-2mm}\sum\limits_{a=0,a\ne v}^{r-1}( \lambda^{t}_{j',v}\hspace{-1mm}-\hspace{-1mm}\lambda^{t}_{j',a}) V_{i',j',u,a} \\
                       \lambda^{t}_{j',v+1} V_{i',j',u,v+1} \\
                       \vdots\\
                       \lambda^{t}_{j',r-1} V_{i',j',u,r-1}
 \end{array}\hspace{-2mm}\right))\\&=&\mbox{rank}(\left(
                    \begin{array}{c}
                        V_{i',u}\\
                        V_{i',j',u,0} \\
                       \vdots\\
                       V_{i',j',u,r-1}
                    \end{array}\right))\\&=&N/r;
\end{IEEEeqnarray*}
Otherwise, $u\ne v$, thus we have\\
\begin{IEEEeqnarray*}{rCl}
                    \mbox{rank}(\left(
                    \begin{array}{c}
                        R'_{i',j'}\\
                       S'_{i',t}A'_{t,j'}
                    \end{array}\right))&=& \mbox{rank}(\left(
                    \begin{array}{c}
                        V_{i',u}\\
                        y_{t,j'}V_{i',u}B'_{t,j'}
 \end{array}\right))\\&=&\mbox{rank}(\left(
                    \begin{array}{c}
                        V_{i',u}\\
                       \lambda^{t}_{j',u}V_{i',u}
                    \end{array}\right))\\&=&N/r.
                \end{IEEEeqnarray*}
\item []Case 2: For $t\in [0,r)$,  $i'\in [rm,(r+1)m)$ and $j'\in [0,rm)$, let  $u=\lfloor \frac{j'}{m}\rfloor$.
If $j'\not \equiv i'$ mod $m$,  we have
\begin{IEEEeqnarray*}{rCl}
&&\mbox{rank}(\left(
                    \begin{array}{c}
                        R'_{i',j'}\\
                        S'_{i',t}A_{t,j'}
                    \end{array}\right))\\&=& \mbox{rank}(\left(
                    \begin{array}{c}
 V_{i',0}+V_{i',1}+\cdots+V_{i',r-1}\\
(V_{i',0}+V_{i',1}+\cdots+V_{i',r-1})B_{t,j'}
                    \end{array}\right))\\&=& \mbox{rank}(\left(
                    \begin{array}{c}
 \sum\limits_{a=0}^{r-1}V_{i',a}\\
 \sum\limits_{a=0}^{r-1}V_{i',j',a,0}B'_{t,j'} \\
                        \vdots\\
 \sum\limits_{a=0}^{r-1}V_{i',j',a,u-1}B'_{t,j'} \\
 \sum\limits_{a=0}^{r-1}V_{i',j',a,u}B'_{t,j'} \\
 \sum\limits_{a=0}^{r-1}V_{i',j',a,u+1}B'_{t,j'} \\
                        \vdots\\
 \sum\limits_{a=0}^{r-1}V_{i',j',a,r-1}B'_{t,j'}
                    \end{array}\right))\\&=& \mbox{rank}(\left(
                    \begin{array}{c}
 \sum\limits_{a=0}^{r-1}V_{i',a}\\
 \lambda^{t}_{j',0}\sum\limits_{b=0}^{r-1}V_{i',j',a,0} \\
                        \vdots\\
 \lambda^{t}_{j',u-1}\sum\limits_{a=0}^{r-1}V_{i',j',a,u-1} \\
 \hspace{-30mm}\sum\limits_{a=0}^{r-1}(\lambda^{t}_{j',u}V_{i',j',a,u}\\ \hspace{8mm}  -\sum\limits_{b=0,b \ne u}^{r-1}(\lambda^{t}_{j',u}-\lambda_{j',b})V_{i',j',a,b}) \\
 \lambda^{t}_{j',u+1}\sum\limits_{a=0}^{r-1}V_{i',j',a,u+1} \\
                        \vdots\\
 \lambda^{t}_{j',r-1}\sum\limits_{a=0}^{r-1}V_{i',j',a,r-1}
 \end{array}\right))\\&=&\mbox{rank}(\left(
                    \begin{array}{c}
 \sum\limits_{a=0}^{r-1}V_{i',a}\\
 \sum\limits_{a=0}^{r-1}V_{i',j',a,0} \\
                      \vdots\\
 \sum\limits_{a=0}^{r-1}V_{i',j',a,r-1}
                    \end{array}\right))\\&=&N/r;
\end{IEEEeqnarray*}
Otherwise,
\begin{align*}
&\mbox{rank}(\left(
                    \begin{array}{c}
                        R'_{i',j'}\\
                        S'_{i',t}A_{t,j'}
                    \end{array}\right))\\=& \mbox{rank}(\left(
                    \begin{array}{c}
 \sum\limits_{a=0}^{r-1}V_{i',a}\\
 \sum\limits_{a=0}^{r-1}V_{i',a}B'_{t,j'}
 \end{array}\right))\\=&\mbox{rank}(\left(
                    \begin{array}{c}
 \sum\limits_{a=0}^{r-1}V_{i',a}\\
 \lambda^t_{j',u}\sum\limits_{a=0}^{r-1} V_{j',a}
                    \end{array}\right))\\=&N/r.
\end{align*}
\item []Case 3: For $t\in [0,r)$, $i'\in [0,rm)$ and $j'\in [rm,(r+1)m)$,
we easily have
\begin{align*}
                    \mbox{rank}(\left(\hspace{-0.3mm}
                    \begin{array}{c}
                        R'_{i',j'}\\
                        S'_{i',t}A'_{t,j'}
                    \end{array}\hspace{-0.3mm}\right))&= \mbox{rank}(\left(\hspace{-0.3mm}
                    \begin{array}{c}
                        V_{i',\lfloor{i'\over m}\rfloor}\\
                        y_{t,j'}V_{i',\lfloor{i'\over m}\rfloor}B'_{t,j'}
                    \end{array}\hspace{-0.3mm}\right))\\&=N/r.
                \end{align*}

\item []Case 4: For $i',j'\in [rm,(r+1)m)$ and $i'\ne j'$, we have
\begin{align*}
&\mbox{rank}(\left(
                    \begin{array}{c}
                        R'_{i',j'}\\
                        S'_{i',t}A'_{t,j'}
                    \end{array}\right))\\=& \mbox{rank}(\left(
                    \begin{array}{c}
                        \sum\limits_{a=0}^{r-1}V_{i',a}\\
                        \sum\limits_{a=0}^{r-1}V_{i',a}B'_{t,j'}
                    \end{array}\right))\\=& \mbox{rank}(\left(
                    \begin{array}{c}
                        \sum\limits_{a=0}^{r-1}V_{i',a}\\
 \sum\limits_{a=0}^{r-1}V_{i',j',a,0}B'_{t,j'} \\
                        \vdots\\
                        \sum\limits_{a=0}^{r-1}V_{i',j',a,r-1}B'_{t,j'}
                    \end{array}\right))\\=& \mbox{rank}(\left(
                    \begin{array}{c}
                        \sum\limits_{a=0}^{r-1}V_{i',a}\\
                        \lambda^{t}_{j',0}\sum\limits_{a=0}^{r-1}V_{i',j',a,0} \\
                        \vdots\\
                        \lambda^{t}_{j',r-1}\sum\limits_{a=0}^{r-1}V_{i',j',a,r-1}
                    \end{array}\right))\\=&N/r.
\end{align*}
\end{itemize}

This finishes the proof after combining (i) and (ii).

\section*{Acknowledgment}
The authors would like to thank the Associate Editor Dr. Parastoo Sadeghi and the three anonymous reviewers for their valuable suggestions and comments, which have greatly improved the presentation and quality of this paper. Jie Li would like to thank Prof. Alexander Barg and Prof. Itzhak Tamo for helpful discussions during his visit at the University of Maryland, College Park.

\ifCLASSOPTIONcaptionsoff
  \newpage
\fi

\begin{IEEEbiographynophoto}{Jie Li} (Member, IEEE) received the B.S. and M.S. degrees in mathematics from Hubei University, Wuhan, China, in 2009 and 2012, respectively, and received the Ph.D. degree from the department of communication engineering, Southwest Jiaotong University, Chengdu, China, in 2017.

From  2015 to   2016, he was a visiting Ph.D. student in the Department of Electrical Engineering and Computer Science, The University of Tennessee at Knoxville, TN, USA.  From   2017 to   2019, he was  a postdoctoral researcher at the Department of Mathematics, Hubei University, Wuhan, China. Since   2019, he has been a postdoctoral researcher at  the Department of Mathematics and Systems Analysis, Aalto University, Finland. His research interests include coding for distributed storage, private information retrieval, and sequence design.

Dr. Li received the IEEE Jack Keil Wolf ISIT Student Paper Award in 2017.
\end{IEEEbiographynophoto}

\begin{IEEEbiographynophoto}{Yi Liu} (Graduate Student Member, IEEE) received the B.S. degree in mathematics
and applied mathematics from Xihua University,
Chengdu, China, in 2014. He is currently pursuing
the Ph.D. degree in information security with
Southwest Jiaotong University. His research interest
includes coding for distributed storage.
\end{IEEEbiographynophoto}

\begin{IEEEbiographynophoto}{Xiaohu Tang} (Senior Member, IEEE)  received the B.S. degree in applied mathematics from
the Northwest Polytechnic University, Xi'an, China, the M.S. degree in applied
mathematics from the Sichuan University, Chengdu, China, and the Ph.D.
degree in electronic engineering from the Southwest Jiaotong University,
Chengdu, China, in 1992, 1995, and 2001 respectively.

From 2003 to 2004, he was a research associate in the Department of Electrical
and Electronic Engineering, Hong Kong University of Science and Technology.
From 2007 to 2008, he was a visiting professor at University of Ulm,
Germany. Since 2001, he has been in the School of Information Science and Technology,
Southwest Jiaotong University, where he is currently a professor. His research
interests include coding theory, network security, distributed storage and information processing for big data.

Dr. Tang was the recipient of the National excellent Doctoral Dissertation
award in 2003 (China), the Humboldt Research Fellowship in 2007
(Germany), and the Outstanding Young Scientist Award by NSFC in 2013
(China). He served as Associate Editors for several journals including \textit{IEEE
Transactions on Information Theory} and \textit{IEICE Transactions on
Fundamentals}, and served on a number of technical program committees of
conferences.
\end{IEEEbiographynophoto}


\begin{thebibliography}{11}

\bibitem{Dimakis} A.G. Dimakis, P. Godfrey, Y. Wu, M. Wainwright, and K. Ramchandran, ``Network coding for distributed storage systems," \textit{IEEE Trans. Inform. Theory,} vol. 56, no. 9, pp. 4539-4551, Sep. 2010.

\bibitem{product} K.V. Rashmi, N.B. Shah, and P.V. Kumar, ``Optimal exact-regenerating codes for distributed storage at the MSR and MBR points via a product-matrix construction," \textit{IEEE Trans. Inform. Theory,} vol. 57, no. 8, pp. 5227-5239, Aug. 2011.

\bibitem{Zigzag} T. Tamo, Z. Wang, and J. Bruck, ``Zigzag codes: MDS array codes
with optimal rebuilding," \textit{IEEE Trans. Inform. Theory,} vol. 59, no. 3, pp. 1597-1616, Mar. 2013.

\bibitem{hadamard} D.S. Papailiopoulos, A.G. Dimakis, and V.R. Cadambe, ``Repair optimal erasure codes through Hadamard designs," \textit{IEEE Trans. Inform. Theory,} vol. 59, no. 5, pp. 3021-3037, May 2013.

\bibitem{transform-IT}  J. Li,  X. Tang, and C. Tian, ``A generic transformation to enable optimal repair in MDS codes for distributed storage systems",   \textit{IEEE Trans. Inform. Theory,} vol. 64, no. 9, pp. 6257-6267, Sept. 2018.





\bibitem{Hadamard-strategy} X. Tang, B. Yang, J. Li, and H.D.L. Hollmann, ``A new repair strategy for the Hadamard minimum storage regenerating codes for distributed storage systems,'' \textit{IEEE Trans. Inform. Theory,}  vol. 61, no. 10, pp. 5271-5279, Oct. 2015.

\bibitem{repair-parity-zigzag} J. Li and X. Tang, ``Optimal exact repair strategy for the parity nodes of the $(k+2,k)$ Zigzag code," \textit{IEEE Trans. Inform. Theory,} vol. 62, no. 9, pp. 4848-4856, Sep. 2016.

\bibitem{transform-ISIT}  J. Li,  X. Tang, and C. Tian, ``A generic transformation for optimal repair bandwidth and rebuilding access in MDS codes",  in \textit{Proc. IEEE Int. Symp. Inform. Theory,} Aachen, Germany, Jun. 2017, pp. 1623-1627.

\bibitem{Long_IT} Z. Wang, T. Tamo, and J. Bruck, ``Explicit minimum storage regenerating codes,"  \textit{IEEE Trans. Inform. Theory,} vol. 62, no. 8, pp. 4466-4480, Aug. 2016.

\bibitem{extend-zigzag} Z. Wang, I. Tamo, and J. Bruck, ``On codes for optimal rebuilding access," in \textit{Proc. 49th Annu. Allerton Conf. Commun.,
Control, Comput.,} Monticello, IL,  Sep. 2011, pp. 1374-1381.

\bibitem{invariant_subspace} J. Li, X. Tang, and U. Parampalli, ``A framework of constructions of minimal storage regenerating codes
with the optimal access/update property," \textit{IEEE Trans. Inform. Theory,} vol. 61, no. 4, pp. 1920-1932, Apr. 2015.

\bibitem{Barg1} M. Ye and A. Barg, ``Explicit constructions of high-rate MDS array codes with optimal repair bandwidth," \textit{IEEE Trans. Inform. Theory,} vol. 63, no. 4, pp. 2001-2014, Apr. 2017.


\bibitem{Barg2} M. Ye and A. Barg, ``Explicit constructions of optimal-access MDS codes with nearly optimal sub-packetization," \textit{IEEE Trans. Inform. Theory,} vol. 63, no. 10, pp. 6307-6317, Oct. 2017.

\bibitem{Sasidharan-Kumar2} B. Sasidharan, M. Vajha, and P.V. Kumar, ``An explicit, coupled-layer construction of a high-rate MSR code with low sub-packetization level, small field size and all-node repair,"   \textit{ arXiv: 1607.07335 [cs.IT]}

\bibitem{YiLiu} Y. Liu, J. Li, and X. Tang,  ``Explicit constructions of high-rate MSR codes with optimal access property over small finite fields," \textit{IEEE Trans. Commun.,}  vol. 66, no. 10, pp. 4405-4413, Oct. 2018.


\bibitem{Goparaju} S. Goparaju, A. Fazeli, and A. Vardy, ``Minimum storage regenerating codes
for all parameters," \textit{IEEE Trans. Inform. Theory,} vol. 63, no. 10, pp. 6318-6328, Oct. 2017.


\bibitem{elyasi2019cascade} M. Elyasi and S. Mohajer, ``Cascade codes for distributed storage systems," \textit{IEEE Trans. Inform. Theory,} to appear.

\bibitem{elyasi2018cascade} M. Elyasi and S. Mohajer, ``A cascade code construction for $(n, k, d)$ distributed storage systems," in \textit{Proc. IEEE Int. Symp. Inform. Theory,} Vail, CO, Jun. 2018, pp. 1241-1245.

\bibitem{li2019systematic} J. Li and X. Tang, ``Systematic construction of MDS codes with small sub-packetization level and near optimal repair bandwidth," in \textit{Proc. IEEE Int. Symp. Inform. Theory,} France, Paris, July 2019, pp. 1067-1071.

\bibitem{balaji2018erasure} S.B. Balaji, M.N. Krishnan, M. Vajha, V. Ramkumar, B. Sasidharan,
and P.V. Kumar, ``Erasure coding for distributed storage: An overview," Sci. China Inf. Sci., vol. 61, Art. no. 100301, Oct. 2018.



\bibitem{Goparaju_bound} S. Goparaju, I. Tamo, and R. Calderbank, ``An Improved Sub-Packetization Bound for Minimum Storage Regenerating Codes,"  \textit{IEEE Trans. Inform. Theory,} vol. 60, no. 5, pp. 2770-2779, May 2014.


\bibitem{alrabiah2019exponential} O. Alrabiah and V. Guruswami, ``An exponential lower bound on the sub-packetization of MSR codes,"  [Online]. Available at: arXiv: 1901.05112 [cs.IT]

\bibitem{Rawat} A.S. Rawat, I. Tamo, V. Guruswami, and K. Efremenko, ``MDS code constructions with small sub-packetization and near-optimal repair bandwidth," \textit{IEEE Trans. Inform. Theory,} vol. 64, no. 10, pp. 6506-6525, Oct. 2018.




\bibitem{Alon} N. Alon, ``Combinatorial nullstellensatz," \textit{Combinat. Probab.
Comput.,} vol. 8, no. 1-2, pp. 7-29, Jan. 1999.


\bibitem{mirsky2012introduction} L. Mirsky, ``An introduction to linear algebra,"  Courier  Corporation, 2012.


%
\end{thebibliography}
\end{document}